\documentclass[format=acmsmall, review=false]{acmart}
\usepackage{thm-restate}
\usepackage{acm-ec-23}
\usepackage{abraces}
\usepackage{easybmat}
\usepackage{booktabs} 
\usepackage{natbib}
\newcommand{\ignorfe}[1]{}
\newcommand{\R}{\mathbb{R}}
\usepackage[ruled]{algorithm2e} 

\SetAlFnt{\small}
\SetAlCapFnt{\small}
\SetAlCapNameFnt{\small}
\SetAlCapHSkip{0pt}
\IncMargin{-\parindent}
\providecommand{\keywords}[1]
{
  \small	
  \textbf{\textit{Keywords---}} #1
}
\newtheorem{theorem}{Theorem}[section]
\newtheorem{lemma}[theorem]{Lemma}

\usepackage{xcolor}

\newtheorem{defn}[theorem]{Definition}
\newtheorem{ex}[theorem]{Example}
\newtheorem{prop}[theorem]{Proposition}

\newenvironment{definition}[1][Definition]{\begin{trivlist}
\item[\hskip \labelsep {\bfseries #1}]}{\end{trivlist}}

\newcommand{\ignore}[1]{}
\setcitestyle{acmnumeric}

\usepackage{abraces}
\usepackage{easybmat}
\usepackage{booktabs} 

\providecommand{\keywords}[1]
{
  \small	
  \textbf{\textit{Keywords---}} #1
}

\title{Wages and Utilities in a Closed Economy}
\author{Sanyukta Deshpande and Milind Sohoni}
\begin{abstract}
The broad objective of this paper is to propose a mathematical model for the study of causes of wage inequality and relate it to choices of consumption, the technologies of production, and the composition of labor in an economy. The paper constructs a {\em  Simple Closed Model}, or an SCM, for short, for closed economies, in which the consumption and the production parts are clearly separated and yet coupled. The model is established as a specialization of the Arrow-Debreu model and its equilibria correspond directly with those of the general Arrow-Debreu model. The formulation allows us to identify the combinatorial data which link parameters of the economic system with its equilibria, in particular, the impact of consumer preferences on wages. The SCM  model also allows the formulation and explicit construction of the \textit{consumer choice game}, where expressed utilities of various labor classes serve as strategies with total or relative wages as the pay-offs. We illustrate, through examples, the mathematical details of the consumer choice game. We show that consumer preferences, expressed through modified utility functions, do indeed percolate through the economy, and influence not only prices but also production and wages. Thus, consumer choice may serve as an effective tool for wage redistribution.    
\end{abstract}

\begin{document}

\begin{titlepage}

\maketitle

\end{titlepage}

\section{Introduction}
 
 As the present millenial edition of the global economy unfolds, many authors and agencies have pointed out several undesirable features that have emerged. These are the paucity of "good" jobs, rising inequality, excessive consolidation, and the possibility of linkages between modern production and consumption processes with climate change. 
 
 This paper deals with the first two problems, {\em viz.}, the allocation of jobs and wages and rising inequality in wage incomes in the economy, and its structural determinants. For this purpose, we construct a simple mathematical model, called the \emph{SCM Model}, which illustrates some of the key features of the dependence of wages on the production and consumption choices within the economy. Our model is divided into two parts: the production sector, determined by a technology matrix $T$ that utilizes $m$ labor classes to produce $n$ goods, where prices are known and revenue maximization is used to determine the quantities of goods produced, labor utilization, and wages. This part of the model assumes prices as a given input that cannot be changed. The second part is the consumption sector, modeled as a Fisher market and a utility matrix $U$. This sector of the economy assumes the production sector as a given, i.e., wages, which are now disposable incomes, and quantities of goods produced, and allocates goods based on disposable incomes held by each labor class, and their expressed utility matrix $U$, and determines prices.
  
This paper demonstrates the connection between $U$, the utility matrix, and the wages obtained by various labor classes, as implemented by $T$. In other words, it traces the connection between personal consumption choices, prices of goods, their production, and finally, the wages received. Furthermore, it shows that having a "private and real" utility $U^t$ and posting or posturing a different $U$ into the economy does indeed alter wages and has the potential to improve both social welfare and the relative welfare for certain classes. The setup of the SCM model then allows us to define the \emph{consumer choice game} (CCG), where the manipulation of $U$ is the strategy, and the relative or total welfare, as measured by the allocation of goods and their utilities according to $U^t$, are the payoffs.

There exists a rich body of literature on the general theory of equilibrium \cite{harsanyi1988general, samuelson1953prices}. The theory seeks to model closed economies and explains the prices of goods and services in a market and how different markets are interconnected. Closed economy models have been studied extensively, starting with Walras's model and the most prominent being the Arrow-Debreu (A-D) market \cite{starr2011general}. Although the existence of equilibrium is proved in models like the A-D market, most of these have seen state policies as the strategy space, not consumer choice \cite{arrow1954existence}. Additionally, many of these models suffer from computational intractability or not providing sufficient ground to model behaviors and strategy spaces as analytical games, which is the primary focus of this study. In the celebrated text, \cite{samuelson1948foundations, backhouse2015revisiting}, Samuelson presents the theory of consumer behavior using ordinal utilities, backed by the principles of "maximizing behavior of consumers and firms" and "stability of equilibrium", which also form the base of our work. Although utility functions as strategies was pointed out \cite{chipman1982samuelson}, it was for the determination of prices and allocations, not wages.

The manipulation of $U$ has also been studied in connection with the impact of advertising on the competition between firms and their profitability. The connection between branding, i.e., managing utilities, and revenues and market shares is well understood \cite{wernerfelt1991brand, delgado2001brand}, but without an understanding of how consumption links to production. Pressure groups utilized this concept to label certain products, such as coffee, as "compliant" with desirable ideals, such as fair wages for coffee-bean pickers \cite{levi2003fair}. Our paper offers a model of how such advice is likely to play out in the economy. Our analysis focuses largely on a closed economy and the study of its modes as functions of the system's parameters.

On the production side, the effects of technology on wages have also been studied in Ricardian and Keynesian economics, mainly in the context of globalization and unemployment \cite{ray1998development, keynes2010economic}. Building on the theory of comparative advantage, other models such as the Heckscher-Ohlin-Samuelson model \cite{bergstrand1990heckscher}, the specific factors model \cite{eaton1987dynamic}, and the most recent Kremer-Maskin model \cite{kremer1996wage}, have analyzed the link between production and technological advancements made available through trade. As long as the division of inputs and outputs of the key blocks of the SCM model are maintained, such variations in the production function may be incorporated into the SCM model. They may help link the effects of globalization and trade with variable consumer choice.

In the computational science community, there has been growing interest in the fields of market equilibria and how strategic agents may manipulate market equilibria under various settings of the economy \cite{branzei2017nash, amanatidis2022allocating, barman2019fair, branzei2014fisher, codenotti2004efficient, branzei2016algorithmic}. Specifically, in connection with our work, the Fisher Market Game \citep{adsul2010nash,garg2017market} has shown that going to the market with postured utilities $U$ (instead of the real ones, $U^t $) can bring rewards in terms of more favorable allocations. The strategic effects of feigning have also been studied under various settings \cite{babaioff2014efficiency, chen2012incentive, chen2011profitable}. However, these models only relate to the consumption side of the economy.  In another related body of work, the Shapley-Shubik Market Game \cite{shapley1977trade}  explaining the formation of prices is studied extensively,  with algorithmic developments on the convergence of proportional dynamics in exchange economies \cite{codenotti2005polynomial,wu2007proportional, zhang2020existence,branzei2021proportional,birnbaum2011distributed,cheung2018dynamics,bei2019ascending}. 

This paper aims to complement these studies by introducing a simple macroeconomic production model helping define a closed economy in its entirety. This allows us to find a consistent set of prices of goods, wages for each labor class, and allocations of goods among classes, resulting in an equilibrium in the economy. Utilizing microeconomic concepts like the Fisher market, we model the effects of consumer choice strategy on wages, production, and allocations. By coupling the co-dependent production and consumption sides of the economy, this work extends the conclusion of the Fisher Market Game to show that consumer choice may indeed be used to derive wages as well.  In classical terms, our work aims to develop and study the dependence of market equilibria on various strategic factors that define the economy, set up via the SCM model and executed via the CCG.

The work draws heavily from existing models discussed earlier, particularly the Arrow-Debreu model and its predecessor, the Fisher model \cite{arrow1954existence}. It also takes inspiration from Sraffa's accounting methods for calculating prices and wages \citep{abraham1979theory}, the use of the theory of value to calculate labor inventory, the use of dual variables for calculating wages, and the use of utility functions to compute the allocation of goods \cite{starr2011general}.

The rest of the paper is organized as follows. In Section \ref{sec:modeling} we describe the SCM Model $\mathcal{H(C,P)}$ as composed of two interconnected systems, the production model ${\mathcal P}$, and the consumption model ${\mathcal C}$. The production model ${\mathcal P}$ has the key parameter $T$, the {\em technology matrix}, and $Y$, the size of individual labor classes. The input variable is the price vector $p$, while the "outputs" are wages $w$ of agents, and the production vector $q$. The consumption model ${\mathcal C}$ consists of the key parameter $U$, the utility matrix, and the inputs as the wages $w$, and production $q$. The "output"  variables are $p$, the prices, and $X$, the allocation of goods to labor classes (or agents). 

We define the market equilibrium in the SCM model, i.e., \textit{SM equilibrium} as $\eta = \{p,q, w, X\}$ with its parameters satisfying both production and consumption system constraints. We also define two global optimization functions $f_{\mathcal C}$ and $f_{\mathcal P}$, which couple ${\mathcal C}$ and ${\mathcal P}$. They allow for a possible {\em tatonnement} as an iterative interaction between ${\mathcal C}$ and ${\mathcal P}$, which we only briefly discuss. 

In Section \ref{sec: smexistence}, we cast the SCM model ${\mathcal H}$ as an Arrow-Debreu (A-D) market. We show the equivalence of SM equilibria, with the equilibria of the corresponding A-D market, proving the existence of equilibrium in the SCM model. Thus, this connects the two concepts and also gives an explicit description of the dependence of the A-D equilibria on the parameters of the economy.

 In Section \ref{sec: ccg}, 
 we define the Consumer Choice Game (CCG) as the game between various labor classes with the expressed utility matrix $U$ as the strategy space and the private utilities and wages as the payoffs. For this, we associate key combinatorial data, {\em viz.}, the Fisher solution forest, with a given market solution. This data determines the regime of interactions between various agents and helps compute the dependence of the payoffs on the matrix $U$ locally. If $U$ changes dramatically, then the market may move to a neighboring Fisher forest, and a new regime, and a new local payoff function.

In Section \ref{sec: correspondence}, we further analyze the combinatorial data arising from an equilibrium point. First, we show through a 2-player example, the decomposition of the strategy space, i.e., the $U$-space into various regions indexed by Fisher forests. This also leads to a correspondence ${\mathcal N}$ between the strategy space and the space of possible pay-offs. In other words, ${\mathcal N}\subseteq \mathbb{R}^2 \times \mathbb{R}^2$. We show that ${\mathcal N}$ is largely a 2-dimensional manifold. 

In Section \ref{sec: roleofstrategy}, we take up the case of a small market of soaps. We utilize the SCM model theory to demonstrate the use of expressed utilities $U$, market structures as described through Fisher solution forests, and lastly, technological updates and labor migration carried out through modifications to $T$ and $Y$, as devices to alter payoffs and allocations in the market.

Finally, in Section \ref{sec: conclusions}, we conclude by pointing out what was achieved, its economic significance, and possible future directions.  

\section{The SCM model} \label{sec:modeling}
In this section, we formally set up the Simple Closed Model (SCM).   Section \ref{subsec:notation} develops the notation and highlights the internal variables of SCM. Section \ref{sec:prodandcons} defines the production and consumption segments of the economy and poses the two as correspondences. Using these, Section \ref{subsec:equilibrium} defines the market equilibrium formally.

\subsection{The basic notation and assumptions}\label{subsec:notation}
\begin{itemize}
\item   
A good is a fixed unit of service or an output of a 
manufacturing plant, made available in a fixed time interval, 
called {\em epoch}. The set of goods, indexed by $j$, is denoted by 
${\mathcal G}=\{ g_1 , \dots , g_n \}$. The column vectors $q$ and $p$ denote the amount of good $g_j $ manufactured in the epoch and its price, respectively.
\item
The entire population in the economy consists of agents
divided into distinct classes, according to their expertise and training.
Set ${\mathcal L}=\{ L_1 ,\ldots, L_m \}$, indexed by $i$,  denotes these classes of labor, with each
class $L_i $ capable of devoting $Y_i $ units (e.g., person-years) of labor in 
every epoch. Thus, $Y$ forms a $m\times 1 $ column vector, indicating labor availability.
\item Each good $g_j $ has exactly one production process or technology 
$T_j$ and each $T_j$
is represented as a column vector $T_j : {\mathcal L}   \rightarrow \R$. 
The class of all such technologies is denoted
by ${\mathcal T}$, which is represented as a $m\times n$-matrix $T$ with
column $T_j $. Thus, in  matrix $T$, the entry $T_{ij}$ is the
amount of labor $L_i$ needed to produce a unit of $g_j$.
\item For each labor class $i$, there is a linear utility function $U_i :\R^n \rightarrow \R$,  that maps $x\in \R^n $, the bundle of goods allocated.  We assume that utilities are the same for all persons in one labor class and are measured in the units - happiness per person per unit weight. For example, $u_{ij}$ is the happiness derived by a person from class $L_i$ by consuming a unit of good $g_j$.\\
As $U_i $ is linear, we write it as $U_i (x)=\sum_{j\in[n]} u_{ij} x_j $, where $u_{ij}\geq 0$ are constants. We assume that for every good $i$ and buyer $j$, there exist buyer $j'$ and good $i'$ respectively, such that $u_{ij'}, u_{i'j} > 0$. We denote the matrix formed by $u_{ij}$ entries as U.

\end{itemize}

Further, without loss of generality, we assume that 
for each labor class, there are technologies 
which utilize them, although the demand for a specific good is generated by the economy. We also assume that the entries
of $T,Y$, and $U$ are all in general position and satisfy no algebraic 
relation (with rational coefficients) amongst themselves.

\noindent
Under this notation, an economy with given specifications $(T,Y,U)$ (and implicit description of $\mathcal{ G, L}$) is parameterized by a {\em non-negative} tuple
$(p,q,w,X)$ in each epoch, where 
\begin{enumerate}

\item $q$ is a $n\times 1$ column vector with $q_j$ 
is the amount of good $g_j $ manufactured,  

\item $p$ is an $n\times 1$ column vector and $p_j $ is the price of each good 
$g_j $, 

\item $w$ is an $1 \times m $ row vector where $w_i $ are the wages 
received by each person in labor class $L_i$ 

\item  $X$ is a $m\times n$-matrix and 
$x_{ij}$, the total amount of good $G_j $ consumed by labor $L_i $ 
\end{enumerate}
The tuple $(p,q,w,X)$ must satisfy additional {\em production} and {\em consumption} conditions specified below. 
\subsection{Production and Consumption} \label{sec:prodandcons}
\subsubsection{Production} \label{subsec: production}
The production space ${\mathcal P}(T,Y; p)$ is 
the collection of all wages $w$ and quantities $q$  such that Eq. \eqref{eq: prodspace1}- \eqref{eq: prodspace3} are satisfied.
\begin{align}
Tq \leq Y \label{eq: prodspace1} \\
q_j \cdot (p_j - (wT)_j) \geq 0  \label{eq: prodspace2} \\
q, w\geq 0 \label{eq: prodspace3}
\end{align}

Note that Eq. \eqref{eq: prodspace1} says that the quantities of goods produced are limited by labor constraints, while Eq. \eqref{eq: prodspace2} says that unprofitable goods are not produced. Consider next, the global maximization function $f_{\mathcal P}$ as $\sum_j p_j q_j$, i.e, revenue maximization and the linear program $P_{\mathcal P}$ and its dual LP $DP_{\mathcal P}$ as given in Eq. \eqref{eq8}.
\begin{equation}
\label{eq8}
P_{\mathcal P}: \  = \left[ \begin{matrix}
\displaystyle \max_q &  p^T q  \\
\textrm{s.t.} 
& T q & \leq & Y \\& \displaystyle \ q_j & \geq & 0 & & \forall j \in [n]
\end{matrix} \right]; \:  \ \ 
DP_{\mathcal P}: \  = \left[ \begin{matrix}
\displaystyle \min_{w} &  w Y  \\
\textrm{s.t.} \ &
\displaystyle  w T &\geq& p^T \\ & \displaystyle \ w_i & \geq & 0 & & \forall i \in [m]
\end{matrix} \right]
\end{equation}
\noindent Utilizing these, we further define the production correspondence as: 
\begin{equation}
    CP(T, Y; p)=\{q, w | q, w \text{  are simultaneous solutions to $P_{\mathcal P}$ and $DP_{\mathcal P}$} \} \label{eq: procorrespondence}
\end{equation}  
By complementary slackness, any $(q', w' )\in CP(T, Y; p) $ satisfies $q'_j ((w'T)_j -pj)=0$ and $(Y_i -(Tq')_i )w'_i =0$. Thus, we see that  if $q'_j >0$ then $(w'T)_j-p_j =0$, and hence $(q',w')\in {\mathcal P}(p)$ as well. Thus, $CP(p)\subseteq {\mathcal P}$. 

Moreover, $w'_i >0$ implies that $(Tq')_i =Y_i $ and more importantly, if $Y_i > (Tq')_i$ then $w_i=0$, i.e., if a particular labor class constraint is not tight then the wages of that class are zero. This implies that the production process conserves money, i.e., $pq' = w'Tq = w'Y$, and the revenue of the producers goes to the labor classes as wages. Also note that the inequality in the dual program, {\em viz.,} $w^T T \geq p^T$, is opposite to the profitability constraint Eq. \eqref{eq: prodspace2}. 

The choice of the above optimization model as a model for production is similar to many earlier economic models, the most popular being the {\em production frontier} model where the feasible region of $P_{\mathcal P}$ from Eq. \eqref{eq8} is used to define the feasible space of the production of goods \cite{ray1998development}. The optimal point is determined either by a family of {\em indifference curves} or an explicit (concave) utility function. The use of optimal production is also the basis of the {\em comparative advantage} argument used by Ricardo to illustrate why trade is beneficial to both parties. When the optimal $\bar{q}$ is a vertex of the production frontier, the wages $\bar{w}$ determined by the theory of marginal productivity matches with our assignment of wages as dual variables. We illustrate the consumption model, and specifically, a trade case in Appendix \ref{ex: prodex1} and \ref{ex: prodex2}, respectively.

\subsubsection{Consumption} \label{subsec:consumption}
We model consumption as a Fisher market \cite{adsul2010nash}. Recall that, in such a market, there are $m$ buyers (i.e., ${\mathcal L}$ labor classes, in our case), and $n$ goods ${\mathcal G}$.  Each class $L_i $ is endowed with money $W_i =w_i Y_i\geq 0$, and each good $G_j $ has quantity $q_j \geq 0$ for sale. The utility function is given by the matrix $U=(u_{ij})$. 

Solution of the Fisher market are {\em equilibrium prices} $p = (p_j)$ with $j\in [n]$ with $p_j q_j >0$ for some $j$, and allocations
$X = (x_{ij} )$ with $i\in [m],j\in [n]$ and $x_{ij}\geq 0$ such that they satisfy the following two constraints:

\begin{itemize}
\item \textbf{Market Clearing}: All useful goods are completely sold and the money of all the buyers is exhausted, i.e.\\
$\forall j \in [n]$,
$ \sum _{i\in [m]} x_{ij} = q_j$ \ and \
$\forall i \in [m]$,
$ \sum _{j\in[n]} p_j x_{ij} = w_i Y_i$ 

\item \textbf{Optimal Goods}: Each buyer buys only those goods which give her the maximum utility per unit of money i.e  if $p_j x_{ij} > 0$, then $\frac {u_{ij}} { p_j} = \max_{j'\in [n], u_{ij'}\neq 0} \frac {u_{ij'}} {p_{j'}} $. \\
This value is denoted as the {\em bang per buck} $bb_i$ and it is defined only for buyers $i$ with $W_i >0$ and then $0<bb_i < \infty$. A generalization to splc utilities is given in Appendix \ref{app: splc}.
\end{itemize}

\noindent

The consumption correspondence is defined as: 
\begin{equation}
    CC(U;q,w)=\{x,p | x,p \text{ satisfy the Fisher market conditions} \} \label{eq: concorrespondence}
\end{equation}  
Given $(x,p)$ as above, the set $X=\{ (i,j)\: | \: x_{i,j}>0 \}$ is called the {\em Fisher graph} of the solution $x$ and records the combinatorial structure of the allocation $x$. It is well known that solutions to the Fisher market are optimal points of the Eisenberg-Gale maximization function, a money-weighted combination of the utilities of the buyers \cite{jain2010eisenberg}. This  allows us to define a global optimization function $f_{\mathcal C}$ for the consumption process. 

As with the production process, money is also conserved in the consumption process- all money with the buyers, i.e., the labor classes, ends up with the producers as revenue. The optimal goods condition is also applicable to goods $j'$ with $q_{j'}=0$. The Fisher market sets the prices of such goods high enough so that no buyer finds it profitable to shift consumption even if an infinitesimal amount of such goods was produced. Hence, this consumption part sees price as the only constraint and no artificial shortage of desirable goods. Also, note that scaling the utilities of a buyer by a fixed constant does not change the equilibrium or the prices, since it is the relative importance of goods for a given agent that determines the relation between the prices of the goods she consumes. We illustrate the consumption process through an example in Appendix \ref{ex: conex1}.


\subsection{Equilibrium in SCM}\label{subsec:equilibrium}

For the economy specifications $(T,Y,U)$, we next utilize correspondences $CP(T,Y)$ and $CC(U)$ to an equilibrium in the SCM model. 
The production happens via $CP$, which takes price $p$ as the input and generates productions and wages $(q,w)$ by maximizing revenues, while $CC$ uses $(q,w)$, to generate the prices via the Fisher market.
\begin{defn}
Given the market $SCM(T,Y,U)$ for the data $T,Y,U$, the tuple $\eta=(q,w,p,X)$ is a {\bf simple closed model equilibrium}, or simply, SM equilibrium, if $(q,w)\in CP(T,Y;p)$ and $(p,X)\in CC(U;q,w)$. 
\end{defn}

The structure of SCM and the definition of SM equilibrium also allow us to naturally define a weak  \emph{tatonnement} process \cite{uzawa1960walras}. The basic objective of the tatonnement is to arrive at an equilibrium $\eta =(p,q,w,X)$ such that (i) $p,X$ are the outputs of the Fisher market if the inputs are $w, q$, on the consumption side, and (ii) $q,w$ are the optimal solutions to $P_{\mathcal{P}}, DP_{\mathcal{P}}$ on input $p$, on the production side. Briefly, the process begins with a candidate solution $\eta $ and checks if $\eta $ is indeed an equilibrium. If not, it alternatively updates the consumption side and the production side by solving the global optimization functions $f_{\mathcal{C}}$, and $f_{\mathcal{P}}$. However, the \emph{tatonnement} does not always converge.  It may alternate between multiple states of the economy if production and consumption do not agree on a common set of active goods and classes but cyclically choose two or more states. We provide more details on its description, and analysis, along with examples in Appendix \ref{app: tatonnement}.

\section{Existence of SM Equilibrium} \label{sec: smexistence}
This section proves the existence of equilibrium $SM(T,Y,U)$. We construct an Arrow-Debreu (A-D) market $AD(T,Y,U)$ from the SCM model $SCM(T,Y,U)$, and show the equivalence of the equilibrium points in the SCM model and market equilibria in the corresponding A-D market. We follow Starr, R. M. (2011), \citep{starr2011general} for the description of the Arrow-Debreu market, and provide specific page numbers for reference\footnote{We assume splc utilities in this section}. 

\subsection{ Arrow-Debreu Market Notation and Conditions} \label{subsec:AD}
The A-D market is given by the following set of specifications. 

\begin{enumerate}
\item The Arrow Debreu (A-D) market consists of $n$ goods, $m$ households {\em viz.}, $H_1 ,\ldots,H_m $, and $l$ firms, {\em viz.,} $F_1 , \ldots , F_l $.  
\item Each firm $F_j $ has a non-empty possible production technology set $Y_j \subseteq \R^n$, where $n$ is the number of goods. A vector $y$ in $Y_j$ represents a technically possible combination of inputs (raw goods) and outputs (produced goods). A negative entry in $y$ represents an input, and a positive entry, an output. We assume that for every $j$, the set $Y_j$ is convex, bounded, closed, and includes $\bar{0}$, the zero vector.  We also assume that there can be no outputs without inputs and there does not exist any way to transform outputs into inputs. (page 115, 165) 
\item For a given price vector $p\in \R^n $, we define $S_j (p)$ as the optimal  supply function of firm $j$  \begin{equation}
       S_j (p) = \{ \ y^*
    |y^* \in Y_j , \ p \cdot y^* \geq p \cdot y \ \ \forall y \in Y_j \ \} \end{equation} 
Then $S_j (p)$ is nonempty and convex for every $p\in \mathbb{R}^n_+$ and upper hemicontinuous as a function on  $\mathbb{R}^n_+ -\bar{0}$. In other words if $(p_i )\rightarrow p$ and $(y_i)\rightarrow y$ are convergent sequences with $y_i \in S_j (p_i)$ for all $i$, then $y\in S_j (p)$. (page 170, 298) 
\item For each household $H_i $, there is a set $X_i \subseteq \mathbb{R}^n_+ $  denoting the set of possible consumption plans of agent $i$. Each $X_i$ is closed, convex, and unbounded above. (page 125) 
\item The economy has an initial endowment of resources which is denoted as $r \in \mathbb{R}^n$. Each household has an initial endowment of goods given by $r^i \in \R^n$ so that $\sum r^i = r$.  Let $Y = \sum_j Y_j$ and $y\in Y$. The vector $y$ is said to be attainable if $y +r \geq 0$. This denotes the possible production of the economy. 
We require that there is a $c>0$ such that $|y|<c$ for all attainable vectors $y$, the existence of which is guaranteed.  
(page 127, 167)
\item Firm $j$'s  profit function $\pi_j :\R^n \rightarrow \R$ is given by:
\begin{equation}
    \pi_j (p) = \max_{y \in Y_j} (p \cdot y ) = p \cdot S_j(p) \end{equation} 
Each household $i$ has shares in firm $j$'s profit given by $\alpha_{ij}\geq 0$ such that  $\sum_{i \in H}\alpha_{ij} = 1$. (page 142)
\item Each agent is endowed with a convex preference quasi-ordering $\geq_i$ on $X_i$. This may be represented by a continuous real-valued function $u_i$ which is the utility function for each $i$. We assume that there is always universal scarcity i.e. for every $x_i \in X_i$, there exists an  $x'_i \in X_i$ so that agent $i$ values it more than $x_i$, i.e., $x_i' >_i x_i $, or simply, $u_i (x_i')>u_i (x_i )$. We also assume that the sets $A_i(x_0) = \{x | x \in X_i , \ x \geq_i x_0 \}$ and $G_i(x_0) = \{x | x \in X_i , \ x_0 \geq_i x \}$ are closed. (page 125, 126, 130)
    \item Income of agent $i$ is defined as:
    \begin{equation}
    M_i (p) = p \cdot r_i + \sum_{j \in F} \alpha_{ij}\pi_j(p)\end{equation}
The budget set of agent $i$ is defined as:
    \begin{equation}B_i (p) =  \{ x | x_i \in\R^n \mbox{ and } p \cdot x \leq M_i(p) |x| \leq c \} \end{equation}
This set should be convex for every $p \in (\R^n_+ -\bar{0})$. (page 175, 181) 
\item Each agent has a demand set given by:
\begin{equation} D_i (p) = \{ x_i|x_i \in B_i (p) \cap X_i , \ \ x_i \geq _i x \ \ \forall \ x \in B_i (p) \cap X_i \}\end{equation}
     Moreover, this set is convex. (page 132)
     \item $B_i(p) \cap X_i$ is continuous (lower and upper hemicontinuous), compact valued, and nonempty for all $p$ in   the unit simplex $\Delta $, i.e., $p\geq 0$ and $\sum_i p_i =1$. Also, $D_i(p)$ is upper hemicontinuous, convex, nonempty, and compact for all $p \in \Delta$. (page 175, 301,303)
\item For all $H_i$, 
\begin{equation}M_i(p) > \inf_{x \in X_i \cap \{ x||x|\leq c\}} p \cdot x \ \forall \  p \in \Delta \end{equation} (page 133)
    
\item Finally, a notation: The excess demand correspondence at prices $p$ is defined  as $Z(p)= \sum_{i=1}^m D(p)- \sum_{j=1}^l S_j (p)- r$. Thus, $Z(p)\subseteq \R^n $ records excess of demand over the supply (including initial resources). (page 147)

    \end{enumerate} 
Arrow and Debreu prove (see Ross, Theorem 24.7\cite{starr2011general}, page 308) that under the assumptions and specifications mentioned above, there exists a competitive market equilibrium, i.e., a sequence of production vectors $(y^*_1 ,\ldots, y^*_l )$, with each $y^*_i \in Y_i$, a price vector $p^* =(p^*_1 ,\ldots ,p^*_n )$, and a set of allocations $x^*_1 ,\ldots ,x^*_m $ with each $x^*_i \in X_i $, all of which satisfy the following conditions:  

\begin{enumerate}
    \item[AD1] Each firm maximizes profits, i.e., $y^{*}_j$ belongs to $S_j(p^*)$ for each $j$. 
    \item[AD2] Each household maximizes utility, i.e., $x_i$ maximizes $u_i(x_i)$ over the set $B_i(p^*)$ i.e. $x_i \ \in D_i(p^*)$. 
    \item[AD3] Prices are non-negative, bounded and not all are zero. Without loss of generality, we can assume that $\sum_j p_j = 1$. 
    \item[AD4] (i) $Z(p^*) \leq 0$. Moreover, (ii) if $ Z_k(p^*)<0$ then $p^*_k = 0 $.
\end{enumerate} 
\subsection{SCM model as an A-D instance}
Given the data $T,Y,U$ for a small market, we shall now build a suitable A-D market. Recall that there are $m$ labor classes and $n$ processed goods, and $Y$ is the $m\times 1$ vector recording the available labor in each class. Also, recall that $T$ and $U$ are $m\times n$ matrices recording the technologies available and the utilities of the classes where $T_{ij}$ refers to the number of labor-units of type $i$ required to produce one unit of good $j$, and $u_{ij}$ denotes the utility of good $j$ to class $i$. We now construct the market $AD(T,Y,U)$. 

\begin{itemize}
\item The total number of goods in $AD$ are $n+m$, viz., $\{ g_1 ,\ldots , g_n , r_1 ,\ldots , r_m \}$, where $r_i $ corresponds to the labor of class $i$.  We call labor inputs `raw' goods.

\item The set of firms in $AD$ is $F=\{ f_1 ,\ldots ,f_n \}$, where $n$ is the number of columns of $T$. The firm $f_j $ produces good $g_j$.
  
\item The number of households is $m$ corresponding to the labor classes, and each agent $H_i $ begins with an endowment $Y_i $ of the good $r_i $ above. 

\item The production function of $f_j $ is $PY_j$ (to differentiate it from the total labor available, viz., $Y_i$) and it arises from the column $j$ of $T$. We define $v_j $ as the $(n+m)$ vector $g_j -\sum_k T_{kj} r_k $ to represent that $T_{kj}$ units of labor type $k$ are used to make one unit of good $g_j$ and the production technology set $PY_j$ as $PY_j =\{ \lambda \cdot v_j | \lambda \in [0,L]\}$ where $L$ is a large number. Thus, the firm $f_j $ may produce $\lambda$ units of good $g_j $ after consuming the necessary resources $\lambda \cdot (\sum_k T_{kj} r_k )$. 

\item Recall that household $H_i $ owns a fraction $\alpha_{ij}$ of the firm $f_j $. We specify these numbers arbitrarily with the prescribed conditions that $\alpha_{ij}\geq 0$ for all $i,j$ and that $\sum_i \alpha_{ij}=1$ for all $j$. The exact numbers will be irrelevant since we will see that in equilibrium, the firms make zero profits. 

\item The space $X_i $ of all possible consumption plans is $\R^{n+m}_+$, the set of non-negative vectors denoting the consumption of the finished goods and the primary goods.  

\item The utility matrix serves to define the continuous real-valued utility function $u_i$ for each agent $i$.  If $x_{ij}\geq 0$ is the amount of good $j$ allocated to agent $i$, then $u_i =\sum_j u_{ij} (x_{ij})$. Utilities are zero for labor unit hours, i.e., $u_{ij}=0$ for $j>n$, as it is only the firms that have any use for labor. Since the utilities are piece-wise linear, it is clear  that the principle of non-satiation holds if, for every $i$, there is a $j$ such that the derivative $u'_{ij}$ is positive, a property which we assume. Also, since $u_{ij}=0$ for $j>n$, i.e., the agents have no use for labor, we may restrict $X_i $ to those elements $x$ with $x_j =0$ for $j>n$, i.e., effectively $X_i =\R^n_+$.   
    
\end{itemize}

This completes the specification of the A-D market $AD$. We verify that the specification satisfies conditions 1-13 of \ref{subsec:AD}. Condition 1 is obvious. Given that $PY_j=\{ \lambda v_j |\lambda \in [0,L]\}$ is a line segment, we may check that conditions 2-3 are satisfied. Condition 4 requires us to compute $S_j (p)$, the optimal supply for a firm $j$, given a price vector $p$. We see that $S_j (p)$ is as given below:
\begin{equation}
 S_j (p)=\left\{ 
\begin{array}{rl}
0 & \mbox{  if $p\cdot v_j<0$} \\
PY_j & \mbox{  if $p\cdot v_j =0$} \\
L\cdot v_j & \mbox{   if $n\cdot v_j >0$} \\
\end{array} \right.
\end{equation}
It is easily seen that $S_j $ is upper hemicontinuous. 

Condition 5 is straightforward. As regards Condition 6, it is easy to find a bound $c\in R$ for the overall production in our economy: one such bound is $m\cdot n$ times the product of the largest labor class and the largest entry in the matrix $T$. Conditions 7-9 are easily satisfied by our piece-wise linear concave utility functions. Condition 10 requires us to compute the optimal bundle $D_i (p)$ for an agent $i$ and argue for the convexity of this set. For an splc utility function, it is easily seen that this is a convex polyhedral set. It is easily seen that this is upper hemicontinuous. Conditions 12-13 are trivial. Thus, our market model $AD(T,Y,U)$ satisfies the A-D conditions of Section \ref{subsec:AD}. 

\subsection{The equivalence of equilibria}

We next show an equivalence between the equilibria of the market $SCM(T,Y,U)$ and the A-D equilibria of $AD(T,Y,U)$. To begin with, we explain the nature of equilibria for $AD(T,Y,U)$ as assured by the A-D result, and set it in the context of our notation for the SCM model. The A-D equilibria is given by the $(n+m)$-vector $P^*$ of prices , the $(n+m)$-vector of productions $(Y^*_j )$ and the $m\times (n+m)$-matrix of allocations $(X^*_{ij})$. 

The price vector $P^*$ specifies prices for all the 
$(n+m)$ goods, i.e., for labor as well as finished goods. Whence, we may denote $P^* =(p^*, w^*)$, where $p^*$ is the $n$-vector recording the prices of goods, and $w^*$ is the $m$-vector recording the wages of the $m$ labor classes. Similarly, the vector $Y^*_j $ must be the multiple $q^*_j v_j $. Finally, the allocations $X^*_{ij}$ for $j>n$ may be assumed to be zero since they do not alter the utility of any household and only (possibly) reduce its budget. Thus, we may assume $X^*$ to be specified instead by $x^*$, an $m\times n$ matrix. Thus, an A-D equilibrium is effectively specified by the tuple $(q^*,w^*,p^*, x^*)$, i.e., in the same format as the SM equilibrium.
\begin{prop} \label{prop:SMAD}
An SM equilibrium $(q^*,w^*,p^*, x^*)$ is also an AD-equilibrium. 
\end{prop}

\noindent
 \begin{proof} Let us prove AD1, i.e., each firm maximizes profits. For this, consider a firm $j$ and the variable $q_j $. By the fact that $w^*$ satisfies $DP_{\mathcal P}$, we know that $w^*T \geq p^*$, i.e., for a column $j$, either the firm is unprofitable or its profit = $p_iq_i - w_iT_iq_i$ is zero. Thus there is nothing to prove when $q^*_j =0$.  On the other hand, when $q^*_j >0$, we know that $(w^* T)_j =p^*_j $ and that the profit is zero, irrespective of $q^*_j$. Thus the firm $j$ maximizes profits. 

Next, let us come to AD2, i.e., each household $H_i $ maximizes its utility given its budget as the only constraint. The budget of the household $j$ is precisely $w^*_j * Y_j $, and this is what the household is endowed with when it enters the Fisher market.  The market clearing condition for the buyer ensures that all money is spent. The maximum bang-per-buck (optimality) condition ensures that, given the prices, only the best possible goods are bought. Thus, while the allocation $x^*_{ij}$ respects conservation constraints, there is no allocation that improves the utility for household $j$, given its budget constraint.

AD3 is obviously ensured. 

For AD4, first, consider AD4(i), i.e., that the supply exceeds demand and the price of a good is non-zero only when supply exactly meets demand. For this, consider first the labor goods $r_i $'s. While the supply is $Y_i $, the condition $Tq^* \leq Y$ ensures that supply exceeds demand. Now, if there is any excess supply of labor class $r_i $, then the dual $DP_{\mathcal P}$ ensures that the wage $w^*_i =0$. For the produced goods $g_j $, the market clearance for the seller ensures that the whole amount $q^*_j $ is allocated, i.e., $\sum_i x^*_{ij}=q^*_j $.  Thus the market clears for every good. Note that  $p^*_j >0$ always holds for an SM equilibrium. This concludes the proof. \end{proof}

\begin{prop}\label{prop: ADSM}
An AD equilibrium $(q^*,w^*,p^*, x^*)$ is also an SM equilibrium. 
\end{prop}

\noindent
 \begin{proof} Let us first examine the condition AD4(i) that $Z(p^*)\leq 0$. For a labor good $r_i $, this precisely states that $Tq^* \leq Y$.  Note that we have chosen a large $L$ to construct $Y_j $, hence we must have $0\leq q^*_j <L$. On the production side, AD4(ii) ensures that $w^*_i (Y_i -(Tq^*)_i)=0$, i.e., $w^*$ exhibits complementary slackness with respect to the conditions $Tq^* \leq Y$. 

Now let us examine AD1 for a firm $f_j $. Since $Y_j =\{ \lambda v_j | \lambda \in [0,L)\}$, the profit of the firm is given by
$q^*_j (p_j -(w^* T)_j)$. Thus if $p_j > (w^* T)_j $ then the optimal value would have been $q^*_j =L$. Since we know that $q^*_j <L$, we have that $(w^* T)\geq p_j $, i.e., the variables $w^*$ satisfy the dual program $w^* T\geq p$ and thus the tuple $(q^* ,w^*)$ indeed satisfy the optimal production condition for SM equilibrium.

Next, let us evaluate $Z(p^*)$ for a good $g_j $. If indeed $Z(p^*)_j <0$, then we are assured $p^*_j =0$. But $Z(p^*)_j <0$ indicates an excess of supply of good $g_j $. Assuming that there is an $i$ such that $u_{ij}>0$ and $w_i >0$ tells us that the allocation $(x_{ij})_j $ is not an optimal allocation since, some extra good $g_j $ is available for free and which is of value to buyer $i$. This tells us that $Z(p^*)_j =0$ for all useful goods and that there is market clearance from the seller side. By the optimality of the consumption plan, i.e., AD2, given the budget, we know that all money is spent, i.e., there is market clearance on the buyer side as well. 

Finally, given the prices $p^*$, for an agent $i$ with $w^*_i >0$, define $b_i =\max_j (u_{ij}/p_j)$, i.e., the bang-per-buck for labor class $i$. Note that if $b_i =\infty$, then the optimal allocation for this class would have been infinite too. Since the AD equilibrium exists and is finite, we must have $p_j >0$ for all useful goods $g_jj$ and $b_i < \infty$. Now, from the optimality constraint AD2, it is clear that whenever
$p_j x_{ij}>0$, we must have $b_i = u_{ij}/p_j$ and that if $x_{ij}=0$, we must have $b_i > u_{ij}/p_j $. This proves the optimality condition of $(p^* , x^*)$ and that $(p^* ,x^*)$ satisfy the optimal consumption requirement of SM. 

This concludes the proof. 
\end{proof}
     
\begin{theorem} \label{thm: existence}
There is an equivalence between the equilibria of the SCM market parametrized by $(T,Y,U)$ and the corresponding A-D market $AD(T,Y,U)$. As a result, for any SCM market given by the above data, there exists an SM equilibrium.
\end{theorem}

\noindent
 \begin{proof} The equivalence follows from Proposition \ref{prop:SMAD}, and Proposition \ref{prop: ADSM}. The existence of an AD equilibrium is guaranteed by Theorem 24.7 of \citep{starr2011general}. \end{proof}

\section{The Structure of SM Equilibria and the Consumer Choice Game} \label{sec: ccg}
 In this section,  we analyze the combinatorial structure of SM equilibria and its connection with market parameters, where we show that these structures are local invariants. This understanding is then used to define the \emph{Consumer Choice Game}, i.e., CCG. 

\subsection{Generic equilibrium and combinatorial data}

For an equilibrium point $\eta $, we define $I(\eta )=\{ i\in [m] | w_i >0 \}$. i.e., the labor classes whose wages are positive. Similarly, we define $J(\eta )=\{ j \in [n]| q_j >0\}$, i.e., the goods which have non-zero production. Similarly, we define $F(\eta )$ as $\{ (i,j)|X_{ij}>0 \}$, i.e., the indices $(i,j)$ such that labor class $i$ has a positive consumption of good $j$.  

We next associate $I(\eta ), J(\eta )$ and $F(\eta )$ as the combinatorial data with equilibrium point $\eta =(p,q,w,X)$ for the parameters $(T,Y,U)$ of the economy. The combinatorial data identify key features of the equilibrium, e.g., the labor classes with non-zero wages, the goods produced, and the Fisher forest, i.e., the price-determining consumption of goods.  We define the notion of `generic-ness', which allows us to construct the equilibrium from its combinatorial data, and to extend such equilibria at a point to its vicinity. 
\begin{definition}
We say that $\eta $ is a generic equilibrium if (i) for $j\not \in J(\eta )$, we have $(wT)_j >p_j $, and (ii) for $(i,j) \not \in F(\eta )$, we have  $u_{ij}/p_j < \max_k u_{ik}/p_k $.
\end{definition}

Let us now fix $T,Y$ and vary $U$ over $\mathcal{U}=\R^{m\times n}$. Given a $U\in \mathcal{U}$, and an equilibrium point $\eta $ with the economy specifications $T,Y,U$, we say that $\eta $ sits over $U$, since it is for this element of $\mathcal{U}$, that $\eta $ was observed. Theorem \ref{thm: T1} relates to the existence of generic equilibria.

\begin{restatable}{theorem}{thmone} 
 \label{thm: T1}
 Let $T,Y,U$ be matrices in general position, i.e., there be no algebraic relationship between the entries, with rational coefficients.   Given an equilibrium $\eta $ over $U$, there are arbitrarily close $U'$ and equilibria $\zeta$ sitting over $U'$ which are generic. Moreover, if $\eta $ has $m$ wage-earning labour classes, i.e., $|I(\eta )|=m$ and $n$ goods produced, i.e., $|J(\eta )|=n$, then the number of connected components ($k$) of the solution Fisher forest $F(\eta )$ is at least $n-m+1$. 
\end{restatable} 
 \begin{proof}
See Appendix \ref{app: T1proof}.
 \end{proof}
 It is an important question if the data $(I,J,F)$ do indeed determine $\eta $, the equilibrium. This is summarized in the Theorem \ref{thm: genericeq}.
     
 \begin{restatable}{theorem}{thmtwo}
 \label{thm: genericeq}
 Again, let $T,Y,U$ be in general position and $\eta $ be a generic equilibrium over $U$ with the combinatorial data $(I,J,F)$, then the parameters of $\eta $, viz., $p,w,q$ are solutions of a fixed set of algebraic equations in the coefficients of $U$.
 For an open set of the parameter space of $\mathcal{U}$, the equilibria, as guaranteed by Theorem \ref{thm: existence}, are generic and have the same combinatorial data as $\eta$. 
\end{restatable} 
 \begin{proof}
See Appendix \ref{app: genericeqproof}.
 \end{proof} 

\subsection{The Consumer Choice Game} \label{sec: stratsinccg}
 We now define the consumer choice game $CCG (T,Y)$,  parametrized by the technology matrix $T$ and the labor inventory $Y$, which are henceforth assumed to be fixed. The players are the labor classes, i.e., ${\mathcal L}=\{ L_1 ,\ldots ,L_k \}$. 
 The strategy space $S_i $ for player $i$ is the utility "row" vector $(u_{i*}) \in \R^n$. These rows together constitute the matrix $U$. This strategy space is denoted by ${\mathcal U}$. We also assume that there is a "real" utility matrix $U^t $ which is used to measure outcomes. The motivation is to advance the implications of the \emph{Fisher market game}, \cite{adsul2010nash}, by defining strategies in the SCM model.
 
 Given a play $U$, the outcome is given by an $\eta (U)=(q,w,y,p,X)$, an equilibrium over $U$ obtained in the SCM model. The payoffs, $b_i (X)= \sum_j (U^t )_{ij}x_{ij}$, i.e., the equilibrium allocations  evaluated by each player on their true utilities, define the preference relations for each player. 
 
 Let us now construct the pay-off functions in the vicinity of a generic equilibrium point $\eta (U)$ with the combinatorial data $(I,J,F)$. We first see that  there is an open set $O_{I,J,F} \subseteq {\mathcal U}$ containing $U$ which has the same combinatorial data $(I,J,F)$.  The exact inequalities defining $O_{I,J,F}$ arise from the requirement that the Fisher forest $F$ have non-negative flows in all edges of $F$, that the edge $(i,j)\not \in F$ has an inferior bang-per-buck, and that $(wT)_j -p_j >0$ for $j\not \in J$. As an example, consider an edge $(i,j)\in F$, and the requirement that the flow in this edge be positive. Now, the flow in this edge is a suitable linear combination of the wages $w$'s, prices $p$'s, and quantities $q$'s. As we have argued before, these, in turn, are smooth functions of the entries of $U$. Thus the condition that flow in the edge $(i,j)$ be positive is the requirement that $f(U)>0$ for a suitable smooth function $f$ on $U$. 
 
 Thus, there is indeed such an open set $O_{I,J,F}$, and the pay-off functions are solutions of algebraic equations in the entries of $U$, the coefficients of which depend on the combinatorial data $(I,J,F)$.   This gives us Theorem \ref{T3} below.
 \begin{theorem}
 \label{T3}
For a generic equilibrium point $\eta (U)$ with the combinatorial data $(I,J,F)$, there is an open set $O_{I,J,F} \subseteq {\mathcal U}$ containing $U$ and a smooth family $\eta'(U')$ of equilibria for each $U'\in O_{I,J,F}$ such that (i) $\eta '(U)=\eta (U)$ and (ii) the combinatorial data for $\eta'(U')$ is $(I,J,F)$.
 \end{theorem}
 The general pay-off function is to be pieced together by such a collection of open sets, indexed by combinatorics. On non-generic $U'$, the equilibrium $\eta (U')$ will have multiple feasible allocations and this determines a correspondence between the strategy space ${\mathcal U}$  and $\R^k $, the pay-off space. Even for a generic $U$, there may be multiple equilibrium points, viz., $\eta_1 ,\ldots, \eta_k $, and each of these will define an analytic sheet of the correspondence over the generic open set. Before we show these characteristics in detail in Section \ref{sec: correspondence}, we demonstrate the theory described so far using an example. 
 
 \subsection{Illustration} \label{subsec: smexample}
We describe an economy $\mathcal{H}_3$ with three labor classes ${\mathcal L}=\{ 1,2,3\}$ and three goods  ${\mathcal G}=\{ 1,2,3\}$ and construct a market equilibrium $\eta$ and the CCG around it, i.e., the payoff functions at $\eta$.  For this equilibrium $\eta $, $I(\eta )=J(\eta )=\{ 1,2,3 \}$. 

\begin{align}
T=\left[ \begin{array}{ccc}
0.5 & 0 & 0 \\ 
0.5 & 2 & 0 \\
0.5 & 4 & 8 \end{array} \right] \:
Y=\left[ \begin{array}{c}
5  \\ 
20 \\
100 \end{array} \right] \: 
U^t =\left[ \begin{array}{ccc}
1.5 & 0 & 0 \\ 
1.5 & 2 & 0 \\
0 & 2 & 1 \end{array} \right] \:
U(\alpha ,\beta)=\left[ \begin{array}{ccc}
1.5 & 0 & 0 \\ 
\alpha & 2 & 0 \\
0 & \beta & 1 \end{array} \right] \label{eq: soapmarket}
\end{align} 

The economy has specifications $T$ (Technology matrix), $Y$ (Labor availability), $ U^t$ (True Utility matrix) and a parametrized matrix $U(\alpha, \beta)$ as given in Eq. \eqref{eq: soapmarket}. Matrix $T$ has columns as technologies $T_1, T_2$ and $T_3$ and rows as the labor classes $L_1, L_2$ and $L_3$. We see that $T_2$ and $T_3$ do not need $L_1$ at all, while $T_1$ utilizes $L_3$ about 16 times more effectively. 
The matrix $U$ is assumed to have two parameters; $\alpha$, to be used by $L_2$, and $\beta$, to be used by $L_3$. Note that $U^t \in U(\alpha,\beta)$. We compute (i) the dependence of the pay-offs on $\alpha $ and $\beta $, and (ii) the sub-domain of $U(\alpha,\beta)$ over which the chosen forest $F=\{(1,1),(2,1),(2,2),(3,2),(3,3)\}$ in Fig. \ref{fig: sec2exampleforest} is the equilibrium forest. 
\begin{figure}[H]
    \centering
   \includegraphics[width=2cm]{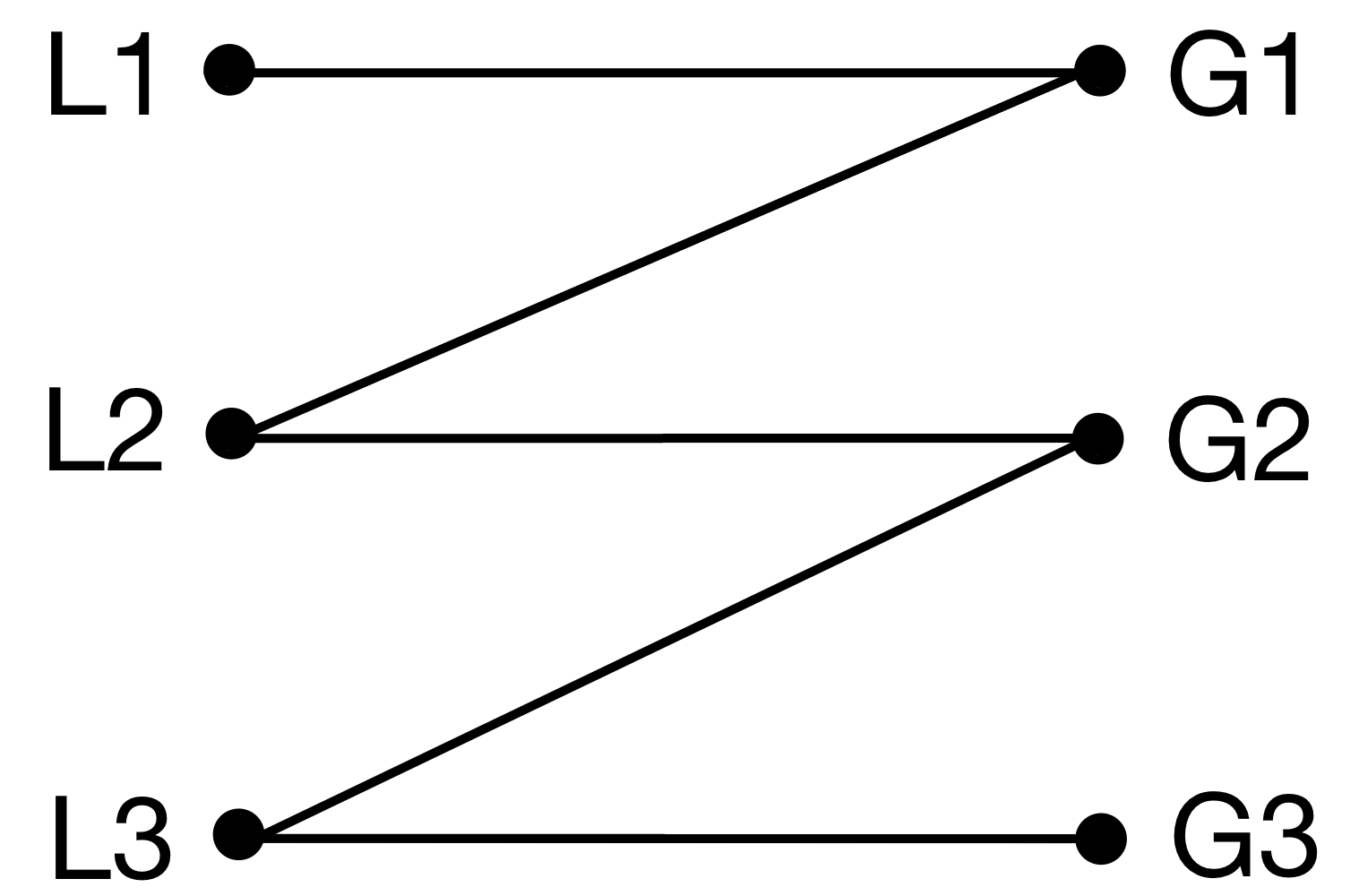}
\caption{Fisher Market Forest} \label{fig: sec2exampleforest}
\end{figure}
The production vector $q$ and wages $w$ follow from $P_\mathcal{P}, DP_\mathcal{P}$, independent of $U$ for fixed combinatorial data $I=J=\{1, 2, 3\}$, as given in Eq. \eqref{eq: exsec2wp}.
\begin{align}
q=T^{-1} Y
=\left[ \begin{array}{c}
10  \\ 
7.5 \\
8.125 \end{array} \right]; \ 
w^T =  [T^{-1}]^T \cdot 
\left[ \begin{array}{c}
p_1 \\ p_2\\ p_3  \end{array} \right] = 
\left[ \begin{array}{c}
2p_1 -p_2/2 +p_3 /8\\
p_2/2 -p_3/4\\
p_3/8  \end{array} \right] \: \label{eq: exsec2wp}
\end{align} 

This describes wages in terms of prices. Note that the support of $U$ gives us utilities for all edges of  the forest, and therefore allows us to compute the Fisher market forest. 
 Using the bang-per-buck relations between prices gives $p_2=\beta p_3$ and $p_1 =(\alpha \beta/2) p_3 $.  Next, the net class wages $W$ and the net revenue collected by the producers $P$ are calculated as: 
 \begin{align}
     W_1 & = y_1 w_1 = 5p_3 * (\alpha \beta -\beta/2 +1/8); & P_1 & =q_1 p_1  = 5p_3 \alpha \beta \notag  \\
     W_2 & = y_2 w_2 = 20p_3 *(\beta/2-1/4); & P_2 &=q_2 p_2  = 15 p_3 \beta /2\\
     W_3 & = y_3 w_3 = 100p_3 /8; & P_3 &= q_3 p_3 = 65p_3 /8 \notag
 \end{align}


Setting $p_3=1$, since the forest $X$ is actually a tree, the conservation equations give us the amount spent  $A_{ij}=p_j *x_{ij}$, and hence allocations $X$, as given in Eq. \eqref{eq: allocsm}.
\begin{align}
   A_{11}&= W_1 =  5p_3 * (\alpha \beta -\beta/2 +1/8) &  x_{11} & = 10 - 5/\alpha + 5/4\alpha \beta \notag \\
   A_{21}& = P_1 -W_1 = 5p_3 * (\beta/2 -1/8) & x_{21} & = 5/\alpha -  5/4\alpha \beta \notag \\
   A_{22}&=  W2-A_{21} = 15p_3 (\beta /2) -5p_3 (7/8) & x_{22} &= 7.5 - 35/8\beta  \label{eq: allocsm} \\
   A_{32}&= P_2 -A_{22} = 5p_3 (7/8) & x_{32}& = 35/8\beta \notag \\
   A_{33}&= W_3 -A_{32}= 65p_3 /8 & x_{33} & = 8.125 \notag 
\end{align}

Finally, the total payoffs $b_i (\alpha,\beta )$ of each labor class $i$ are the net ``true'' utilities received from their total consumption, as given in Eq. \eqref{eq: smpayoffs}.
\begin{align}
b_1 (\alpha ,\beta )&= 1.5x_{11} (\alpha ,\beta) \notag \\
b_2 (\alpha ,\beta )&=1.5x_{21} (\alpha ,\beta) + 2 x_{22} (\alpha ,\beta) \label{eq: smpayoffs} \\
b_3 (\alpha ,\beta )&= 2x_{32} (\alpha ,\beta) + x_{33} (\alpha ,\beta) \notag 
\end{align} 
The conditions that wages, prices, and allocations, i.e., the money flows $x_{ij}$, need to be positive, give us the following constraints: (i) $8\alpha \beta -4\beta +1 > 0$, (ii) $4\beta -1  >  0$ and (iii) $ 60\beta -35  >  0$. Under these conditions that define an open set, forest $F$ arises as the equilibrium forest, giving the fixed combinatorial data $(I,J,F)$. Note that the payoffs in Eq. \eqref{eq: smpayoffs} are also valid whenever $\eta $ remains the market equilibrium, i.e.,  where $\alpha $ and $\beta$ are in the open set. Thus, we have an $SM$-equilibrium $\eta(\alpha,\beta)=(w(\alpha , \beta),q(\alpha ,\beta ),p(\alpha ,\beta ),X(\alpha,\beta))$, parametrized by $\alpha$ and $\beta$.

This illustrates that the local combinatorial data is sufficiently explicit to enable the computation of the pay-off functions. Moreover, significant benefits may accrue to players if they utilize the freedom of posturing their utility functions.

\section{Correspondence associated with CCG} \label{sec: correspondence}

In this  section, we illustrate the connection between the strategy space $\mathcal{U}$ and the combinatorial data of Fisher forests. We utilize a small example to demonstrate the theory and refer to Appendix \ref{app: stratsinccg} for a general market scenario.

We consider a two-class economy with  specifications $T$ (Technology matrix), $Y$ (Labor availability), $ U^t$ (True Utility matrix), and ${\mathcal U}=\R^{2\times 2}$ (Strategy matrix) as given in Eq. \eqref{eq: ccgexeconomy}. Since Fisher solutions do not change with independent row scaling of the utility matrix, the description of $U$ suffices in representing $\{\mathcal U\}$, with $\alpha= \frac{u_{11}}{u_{12}}$ and $\beta= \frac{u_{21}}{u_{22}}$.
\begin{align}
    T=\left[ \begin{array}{ccc}
0.25 & 0  \\ 
0.25 & 1  \end{array} \right] \: 
Y=\left[ \begin{array}{c}
2  \\ 
4 \end{array} \right] \:
 U^t= \left[ \begin{array}{cc}
  1   & 1 \\
  1   & 1 
\end{array} \right] \:
U= \left[ \begin{array}{cc}
  \alpha   & 1 \\
  \beta   & 1 
\end{array} \right] \label{eq: ccgexeconomy} \end{align} 
Assuming $0 < \alpha, \beta < \infty$, we now cover the entire strategy space with the finite collection of Fisher Forests. This enable us to define payoff functions for the whole strategy space. 
In a two-class, two-goods economy, there are nine possible ways to allocate the goods produced among the classes. These include six forests and one cycle, where both classes participate; and two forests where only one class is active. Out of these nine possibilities, the inputs $(T, Y)$ allow for six distinct combinatorial structures while satisfying the production and consumption side constraints, as given in Fig. \ref{fig: ccgexforests}.
\begin{figure}[H]
    \centering
\includegraphics[width=10cm]{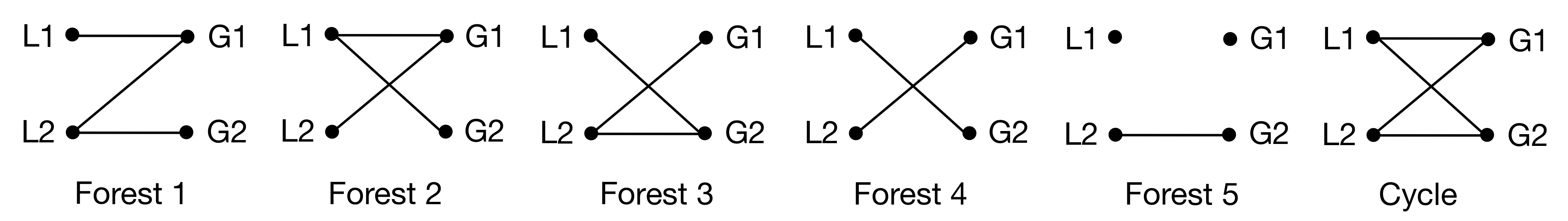}
\caption{Feasible combinatorial structures for equilibrium $SM(T,Y,U)$, with for $J(\eta) = \{1,2\}$. }
    \label{fig: ccgexforests}
\end{figure}
 We see that the actual combinatorial data $(I,J,F)$ and the SM equilibrium values $(p,q,w,X)= SM(T,Y,U)$ depend on the choice of $\alpha$ and $\beta$.   As guaranteed by Theorem \ref{thm: genericeq},  there is an open set for each forest, a "zone", of the $\alpha$-$\beta$ space over which the forest is the Fisher forest of the equilibrium. For fixed combinatorial data $(I,J,F)$, the corresponding zone-defining conditions can be derived from the positivity of wages and allocations, and the maximum bang-per-buck conditions, as seen in effect in Section \ref{subsec: smexample}. Conditions specific to the feasible combinatorial structures are given in Table \ref{tab: forestdescriptions}, with a complete description of corresponding $(p,q,w,X)$ in Appendix \ref{app: foresttable}. 
 
\begin{table}[h]
\centering
\small
\begin{tabular}{|l|l|l|l|l|l|l|}
\hline &&&&&&\\[-0.7em]
 & Forest-1                                                                                                      & Forest-2                                                                                                     & Forest-3         & Forest-4   & Forest-5    & Cycle                                                                                       \\ \hline &&&&&&\\[-0.7em]
     Zones  & $\alpha \geq \beta > 1/4$  & $\beta \geq \alpha > 1/2$ & $1/4 < \beta< 1/2$ and  $\beta \geq \alpha$ & $ \alpha \leq 1/2 \leq \beta $ & $\beta \leq 1/4$ & $ \alpha=\beta $
     \\ \hline
\end{tabular}
\caption{$SM(T,Y,U)$ zone descriptions for the structures in Fig. \ref{fig: ccgexforests}} \label{tab: forestdescriptions}
\end{table} 
\normalsize

We also classify the generic  equilibrium points here. Recall that $\eta $ is a generic equilibrium if (i) for $j\not \in J(\eta )$, we have $(wT)_j >p_j $, and (ii) for $(i,j) \not \in F(\eta )$, we have  $u_{ij}/p_j < \max_k u_{ik}/p_k $. Therefore, we see that an equilibrium point belonging to, say zone-1, is generic if and only if $\alpha > \beta$. We can define an open set for zone-1, {\em viz. }, $\alpha > \beta > 1/4$, comprising only of generic equilibrium points. Similarly, for zone-4, the open set would be $ \beta > 1/2 > \alpha $. 
\begin{figure}
    \centering
    \includegraphics[width=4.5cm]{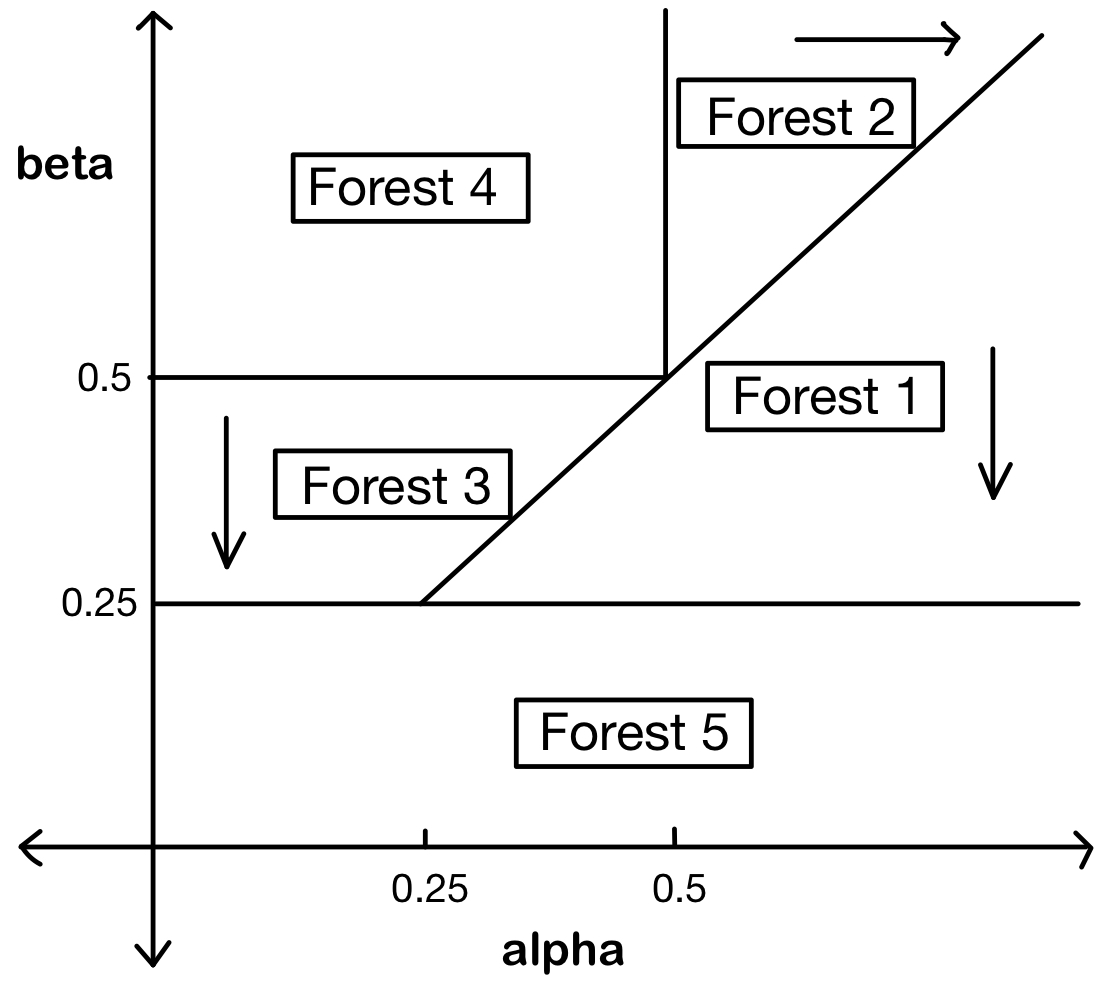}
    \caption{The equilibrium $SM(T,Y,U(\alpha,\beta))$ categorized into multiple zones, as a function of $(\alpha,\beta)$.}
    \label{fig: zones}
\end{figure}
Finally, we may summarize this discussion through Fig. \ref{fig: zones}, where the five zones are marked corresponding to the forests. The arrows  indicate the direction of optimization for the dominant class i.e. the class that has control over the ratio $\alpha$ or $\beta$. The interior of each zone is an open set consisting of generic points and the boundaries correspond to the non-generic equilibrium points.

Fig. \ref{fig: zones} shows that by changing $\alpha$ or $\beta$, it is possible to transit from one forest to another, by crossing the non-generic  (boundary) points where both forests are feasible.   It is shown earlier that the set of allocations in the Fisher market is hemicontinuous with respect to initial endowments and utility functions \cite{vazirani2011market} . Here, we show that although multiple allocations are possible at such points, utilities on the boundary points are bounded by the utility limits of the forests on both sides. In other words, allocations and utilities at the transitions are convex combinations of the boundaries of those obtained in the adjoining zones. To make this precise, we make the following definition.
\begin{definition}
 Let $x\in \mathcal{U}$ be a point on the boundary of two zones, say Zone A and Zone B and let $\eta =(p,q,w,X)$ be a typical point above $x$, i.e., $\eta $ is an equilibrium for the parameter $x$. Let $\mathcal{X}(x)$ be the collection of all allocations of equilibria above $x$ in the U-space. We say that $x$ is a \textit{manifold point} if the set $U_i (\mathcal{X}(x))$ is a  bounded interval and its bounds are obtained as the limits  $\lim_{q \rightarrow x} U_i (X(q))$ and $q\in$ Zone A and $q \in$ Zone B. 
\end{definition}

For a point $x$ located on the line $\alpha = \beta$, with $\alpha (x) = \beta (x) = \mu > 0.5$, i.e. on the transition between Zone 1 and Zone 2, the corresponding equilibrium $\eta (x) = (p,q,w,X)$ has: (i) $I(x) = J(x) = \{ 1,2\}$, and (ii) $F(p) = \{(1,1),(1,2),(2,1),(2,2)\}$, i.e. a cycle. Furthermore, there exists a set of possible allocations  $\mathcal{X}$ at $\eta (x)$, and the bounds on $U_i (\mathcal{X})$ are precisely those achieved as $\lim_{q\rightarrow x} U_i (X(q))$ for $q$ in zone 1 and 2, making it a manifold point. Similarly, we find that $x$ located on $1/4 < \alpha = \beta < 1/2$, i.e., between Zone 1 and 3, is a manifold point. For the points on the boundaries with Forest 4, however, we see that the limits $\lim_{q\rightarrow x} U_i (X(q))$ are equal from both zones, thus giving a unique allocation $X$.  On the other hand, we see that the points on the boundary $\beta = 1/4$ are not manifold points. This is because, for each point on this line,  there is a unique allocation $X$ leading to a unique $U_i (X)$ which does not equal $\lim_{q\rightarrow x} U_i (X(q))$ for $q$ in zone 1 or 3.

Further corresponding to the interior, i.e., generic points of the zones, we have payoffs $b_1,b_2$ defined uniquely, which are continuous functions of $\alpha, \beta$, as guaranteed by Theorem \ref{T3}. Moreover, these are invertible functions on their restricted domain of $\alpha, \beta$, which makes the sets of possible payoffs $b^p_1$, $b^p_2$ open. It follows that the correspondences $\mathcal{N}_i = (\alpha, \beta, b^p_1, b^p_2)$ are open for each forest $i$. Moving further, we see that for points on $\alpha = \beta$, associated with a cycle, the corresponding  $b^p_1$ and $b^p_2$ belong to open sets. For example, on the boundary of forests 1 and 3, the open set  $U^p_1$ is given by $1/4 < \alpha  < 1/2$ and  $(8\alpha-2,  8- 2/\alpha)$.  The region looks like half of a parabola, bounded from all sides. Thus, the interior of every possible solution or zone is open and the neighborhood of each point is homeomorphic to open subsets of $\mathbb{R}^2$. For forests, the homeomorphism is given by the inverse of utility functions and for the cycle, the sets are open in $\mathbb{R}^2$.

Finally, we claim that the boundaries of the forests and the cycle form 1-dimensional entities which serve as boundaries to the described 2-dimensional manifolds. Each boundary can be given by a unique equation in $\mathbb{R}^2$.  On line $\alpha=\beta$, there are two boundaries, one coming from $\alpha < \beta $ and another from $\alpha > \beta$. In between these two, a 2-dimensional plane is situated on each part of the segment $\alpha = \beta$, i.e. on the boundary of forest-1 and 3 and forests 1 and 2. When considering the closures of the open sets, we see that correspondences intersect along these boundaries. 

We therefore establish that in the region $\beta > 1/4$, payoff function $b_i$ is a 2-dimensional manifold with boundary, consisting of all `manifold' points. Moreover, a correspondence $\mathcal{N}$ = $(\alpha, \beta, b^p_1, b^p_2)$ can be defined between the strategy space and the payoffs space. A general version of markets is dealt with in Appendix \ref{app: stratsinccg}, where we argue that the same results follow. \\
\noindent
\textbf{Strategic Analysis:}
As shown in Fig. \ref{fig: zones}, class-2 is dominant in the sense that it has a strategy to become the only active class in the economy, just by reducing $\beta $. However, it is not in its interest to completely drive out class 1. If the state is in any zone with both classes active, $\beta$ can be decreased to reach zone 5, where class-2 gets utility 4. But, due to a discontinuity in the utility function, for any $\alpha$, class-2 achieves the highest payoff (getting arbitrarily close to 10), whenever its $\beta$ approaches 1/4, but is,  strictly more than 1/4. Thus, the best strategy for class-2 is to keep class-1 active and pose $\beta$ as close to, but greater than 1/4. Technically, we see that the  discontinuity of the utility functions results in the non-existence of Nash equilibria.
 
Although this economy rules out the possibility of a Nash equilibrium by making one class clearly dominant,  the general scenario has a possibility of the existence of Nash equilibria and is discussed in Appendix \ref{app: stratsinccg}.

\section{The role of strategy} \label{sec: roleofstrategy}
In this section, we utilize the SCM model to discuss the case of a specific example - the small market of soaps. We demonstrate the use of $U_i$, the utility vector of an agent or labor class $i$, as devices to change allocations. We also illustrate how $T$, the set of technologies, and $Y$, the availability of labor determine $X$, the market structure, and the payoffs in the small market.
\subsection{The market of soaps}

We consider a society that produces a single product, such as household soaps. The labor classes involved in the production of this product may be classified as follows: $L_3$, a low-technology worker in the factory, or an individual producer at the household scale; $L_2$, an intermediate-technology worker who can adapt processes to address consumer choice, coordinate and improve production and distribution processes or is able to lead small and medium enterprises; $L_1$, a high-end technology worker who can bring international investments and machinery, can produce at a very large scale, and can run elaborate distribution and marketing, networks. The matrices $T,Y,U^t$ are as given in Section \ref{subsec: smexample} and are reproduced again in Eq. \eqref{eq: soapmarket2}. This is a simplified representation of the data on page 16, \cite{awdate}, a report on soap manufacturing in India. 

\begin{align}
T=\left[ \begin{array}{ccc}
0.5 & 0 & 0 \\ 
0.5 & 2 & 0 \\
0.5 & 4 & 8 \end{array} \right] \:
Y=\left[ \begin{array}{c}
5  \\ 
20 \\
100 \end{array} \right] \: 
U^t =\left[ \begin{array}{ccc}
1.5 & 0 & 0 \\ 
1.5 & 2 & 0 \\
0 & 2 & 1 \end{array} \right] \:
U(\alpha ,\beta)=\left[ \begin{array}{ccc}
1.5 & 0 & 0 \\ 
\alpha & 2 & 0 \\
0 & \beta & 1 \end{array} \right] \label{eq: soapmarket2}
\end{align} 

The technology matrix $T$ lists the ways of manufacturing soaps in this economy. The last column $T_3 $ uses only the $L_3$ labor class, while $T_2$ is a boutique manufacturing process that utilizes both $L_2$ and $L_3$ labor classes but with a reduced requirement of $L_3$ and a less than proportionate the use of $L_2$. $T_1$ is the most efficient method in terms of manpower. It uses a half-unit of $L_1$ to substantially reduce the number of both $L_2$ and $L_3$. In real life, $T_1$ encompasses various activities such as packaging, transport, advertising, rents on expensive machinery, and other activities related to mass production and distribution.

The payoff functions $b_1, b_2$ and $b_3$ of individual labour classes as functions of the strategy vector $U(\alpha , \beta)$, are also as in Section \ref{subsec: smexample} and are reproduced below:
\begin{align}
b_1 (\alpha ,\beta )&= 1.5x_{11} (\alpha ,\beta) \\
b_2 (\alpha ,\beta )&=1.5x_{21} (\alpha ,\beta) + 2 x_{22} (\alpha ,\beta)  \\
b_3 (\alpha ,\beta )&= 2x_{32} (\alpha ,\beta) + x_{33} (\alpha ,\beta)
\end{align} 
Note that, the allowed flexibility in $U$ is restricted to $\alpha$ and $\beta$, merely for the purpose of illustration and easy computability of the payoff functions. 
\subsubsection{Utility amendments:}
The payoff functions may be used to construct a finite strategy bimatrix game for the labour classes $L_2 $ and $L_3$. We choose $3$ values for $\alpha$, viz., $\alpha_1 =1,\alpha_2 =1.5,\alpha_3 =3$ and $3$ values for $\beta$, viz., $\beta_1 =1,\beta_2=2,\beta_3 =3$, which lie within the open set of valid values, derived in Section \ref{subsec: smexample}. 
For nine combinations of $\alpha, \beta$, we record the net payoffs $b_i (\alpha,\beta )$ for all three classes in Table \ref{tab: soapmarketutilities}. The center of each table, i.e., the values $b_i (\alpha_2,\beta_2 )$ correspond to $U=U^t$, i.e., true expression of utilities.
\begin{table}[H]
\small
    \centering
$ \begin{array}{|c||c|c|c|}\hline 
L_1 & \beta_1 & \beta_2 &\beta_3 \\ \hline \hline 
\alpha_1 & 1.875 & 1.688 & 1.625 \\ \hline 
\alpha_2 & 2.250 & {\bf 2.125} & 2.083 \\ \hline 
\alpha_3 & 2.625 & 2.563 & 2.542 \\ \hline \end{array} \: \: \: 
\begin{array}{|c||c|c|c|}\hline 
L_2& \beta_1 & \beta_2 &\beta_3 \\ \hline \hline 
\alpha_1 & 0.594 & 0.859 & 0.948 \\ \hline 
\alpha_2 & 0.500& {\bf 0.750} & 0.833  \\ \hline 
\alpha_3 & 0.406 & 0.641 & 0.719 \\ \hline \end{array} \: \: \:
\begin{array}{|c||c|c|c|}\hline
L_3 & \beta_1 & \beta_2 & \beta_3 \\ \hline \hline 
\alpha_1 & 0.169 & 0.125 & 0.110 \\ \hline 
\alpha_2 & 0.169 & {\bf 0.125} & 0.110 \\ \hline 
\alpha_3 & 0.169 & 0.125 & 0.110 \\ \hline \end{array} $
    \caption{Net payoffs $b_i(\alpha, \beta)$  for classes $L_i$, $i\in \{1,2,3\}$, with "true" values in bold}
    \label{tab: soapmarketutilities}
\end{table}
\normalsize

In the true state, there is a large disparity in the payoffs received by the three labor classes. $L_2$ earns 6 times that of $L_3$ and $L_1$ earns 2.833 times that of $L_2$ and 17 times that of $L_3$. It is also evident that there is a significant benefit to be gained by reporting $U\neq U^t$, i.e., modified utility vectors.
In fact, $L_3$ can increase its payoffs by $35\%$ by reducing $\beta_2$ to $\beta_1$, $L_2$ can then minimize its losses by $21\%$ by reducing $\alpha_2$ to $\alpha_1$, and this causes a loss of $12\%$ to $L_1$. This implies that $L_3$ should devalue good $g_2$ in comparison to $g_1$, while $L_2$ devalues $g_1$ in comparison to $g_2$. In other words, the Nash equilibrium for this game is $(\alpha_1,\beta_1 )$. 
\subsubsection{The boutique market structure:}
We now consider an alternate market structure, with $F_b =\{ (1,1),(1,2),(2,1),(3,1),(3,3)\}$, as shown in Figure \ref{fig: boutique}. 
\begin{figure}[H]
    \centering
   \includegraphics[width=2.2cm]{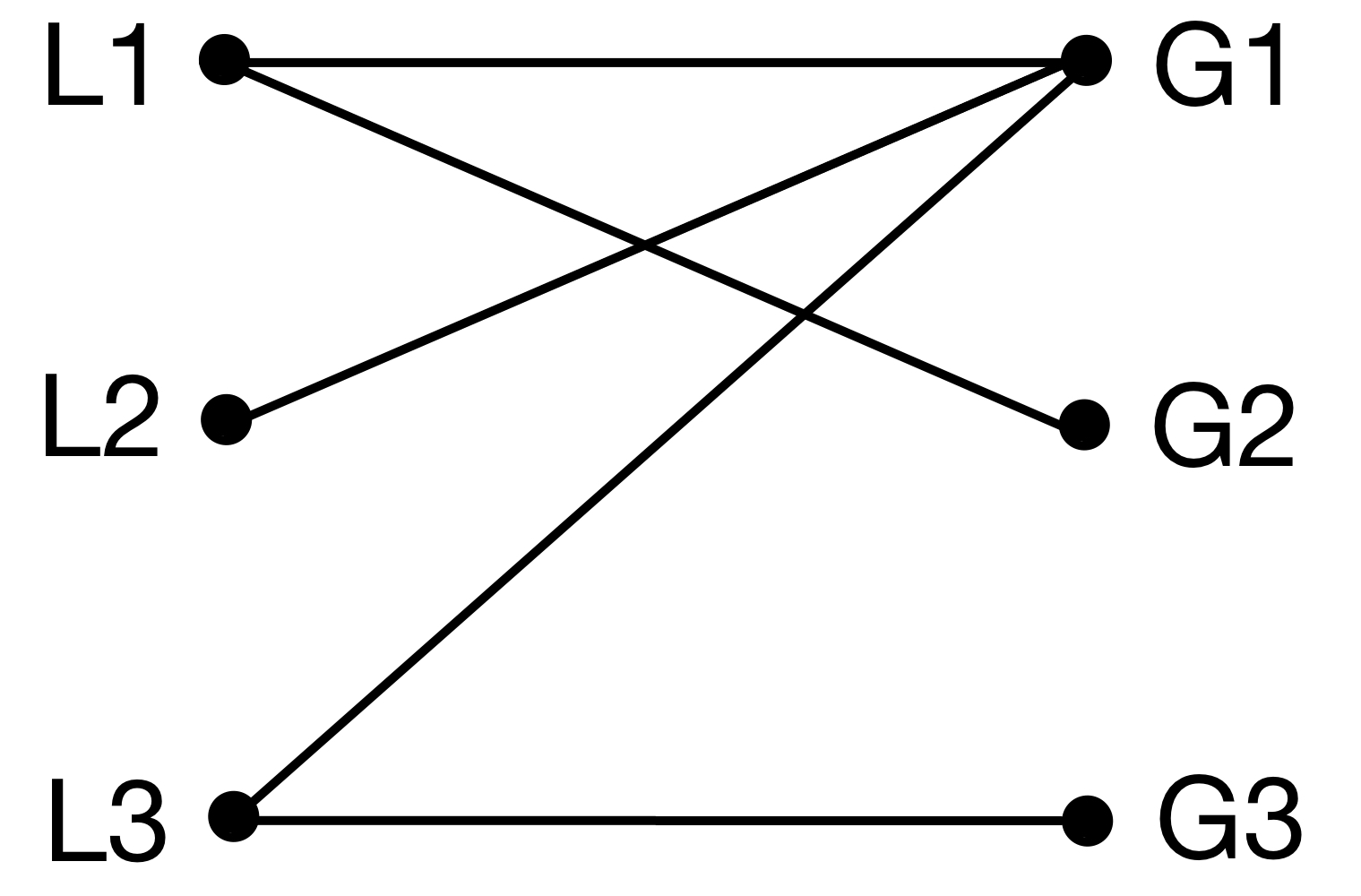}
\caption{Fisher Market Forest: boutique market structure} \label{fig: boutique}
\end{figure}
Here the production of the boutique soap, $g_2$, is completely consumed by $L_1$. All classes consume the mass-produced soap, $g_1$, while only $L_3$ consumes its own household production, $g_3$. We define $U_b (\alpha, \beta )$ as the utility function for this consumption structure, where the ``true" utility matrix $U^t$ is $U_b (\frac{4}{3}, \frac{2}{3})$, as in Eq.  \eqref{eq: soapmarketstructures}.
\begin{equation}
U_b (\alpha , \beta) =\left[ \begin{array}{ccc}
1 & \alpha & 0 \\  
1 & 0 & 0 \\
1 & 0 & \beta \end{array} \right] \: \: \: \: \: \: 
U^t =\left[ \begin{array}{ccc}
1 & 4/3 & 0 \\ 
1 & 0 & 0 \\
1 & 0 & 2/3 \end{array} \right] \label{eq: soapmarketstructures}
\end{equation}

The price vector corresponding to $U_b (\alpha , \beta)$ would be $p(\alpha , \beta)=[1, \alpha , \beta]^T$. Since the technology matrix $T$ and the labor composition of the population, $Y$, remain the same, the quantity of each good produced remains unchanged at $q=[10, 7.5, 8.125]^T$.

It turns out that this market structure is not feasible for $U(\frac{4}{3},\frac{2}{3})$, i.e., at ``true'' prices, this consumption pattern will not be observed. However, it is seen for $U(1,\frac{2}{3})$, where the boutique soap $g_2$ is valued (and priced) the same as the mass-produced $g_1$. It is also seen at $U(1,1)$, where $L_3$ chooses to value its household production $g_3$ as equal to $g_1$. The corresponding payoffs are given in Table \ref{tab: soapmarketstructtab}.
\begin{table}[H]
\small 
    \centering
    $\begin{array}{|c||c|c|c|}\hline 
(\alpha ,\beta) & b_1 & b_2 & b_3 \\ \hline \hline 
(1,2/3) & 3.124 &0.503 & 0.125 \\ \hline 
(1,1)  & 3.187 &0.375 & 0.147\\ \hline 
\end{array}$
    \caption{Total payoffs for all the classes, under different settings of $(\alpha, \beta)$}
    \label{tab: soapmarketstructtab}
\end{table} 
\normalsize
\vspace{-0.5cm}
The reasons for the unsustainability of ``true'' prices in $X_b$ are straightforward: Given that all boutique production, i.e., of $g_2$, must be consumed by $L_1$, either the amount produced, i.e., $q_2$, must be small or the price $p_2$ must be low as compared to $p_1$. This is easily verified by setting the labor composition to $Y_b$ equal to $[5, 10, 100]^t$, leading to a lower production of $g_2$. With this, allocation $X_b $ becomes feasible and the payoff equal $b=[2.129, 0.750, 0.128]$.
\subsubsection{Technological updates:}
Here, we examine the impact of changes in technology, $T$, while preserving its hierarchical structure of labor classes. We introduce a new technology matrix, $T'$, which differs from the original matrix, $T$, as given in Eq. \eqref{eq: soapmarkettech}. 
\begin{equation}
T=\left[ \begin{array}{ccc}
0.5 & 0 & 0 \\ 
0.5 & 2 & 0 \\
0.5 & 4 & 8 \end{array} \right] \: \: \: \: \: \: 
T'=\left[ \begin{array}{ccc}
0.5 & 0 & 0 \\ 
0.5 & 2 & 0 \\
1 & 4 & 8 \end{array} \right]
 \label{eq: soapmarkettech}
\end{equation} 
This change may occur by internalizing some of the informal labor required in the life-cycle of good $g_1$, which is mass-produced soap. This could include incorporating additional garbage collection and disposal cycle for the production good $g_1$, which is not done for good $g_2$ and $g_3$. By internalizing this process, $T'$ represents a possible new technology matrix. Table \ref{tab: soapmarkettechtab} shows the comparison of the new market equilibrium, i.e., $SM(T', Y, U)$ with that of the original one, i.e., $SM(T, Y, U)$.
\begin{table}[H]
\small
    \centering
$  \begin{array}{|c||c|c|c||c|c|c|}\hline 
 & \multicolumn{3}{|c|}{T} & \multicolumn{3}{|c|}{T'} \\ \hline
\text{production} & 10& 7.5 & 8.125 & 10 & 7.5 & 7.5 \\ \hline 
\text{payoffs} & 2.125 & 0.75 & 0.125 & 2 & 0.75 & 0.125 \\ \hline 
\end{array} $
    \caption{Comparison of the goods production and class payoffs, with technologies $T$ and $T'$ }
    \label{tab: soapmarkettechtab}
\end{table} 
\normalsize
\vspace{-0.5cm}
Table \ref{tab: soapmarkettechtab} shows that while the production of good $g_3$ has decreased, it is the payoffs of labor class $L_1$ that have decreased while the production and payoffs of other labor classes and goods remain unchanged.
\subsubsection{Labor migration:}
Finally, we examine the case with changes in $Y$, representing the relative strengths of the labor classes. This may occur, for example, through education. For instance, we can consider changes in $Y_{21}$ and $Y_{32}$ obtained by moving 5 units from $L_2$ to $L_1$ and $L_3$ to $L_2$ respectively, as given in Eq. \eqref{eq: soapmarketY}.
\begin{equation} Y=\left[ \begin{array}{c}
5  \\ 
20 \\
100\end{array} \right] \: \: \: Y_{21} = \left[ \begin{array}{c}
10  \\ 
15 \\
100 \end{array} \right] \: \: \: Y_{32} = \left[ \begin{array}{c}
5  \\ 
25 \\
95 \end{array} \right] \label{eq: soapmarketY}
\end{equation}
If $U^t$  is being asserted by every labor class, then while the quantities of production may change with $Y$, the payoffs of the various classes remain unchanged. However, if there is a deviation from the expected values of $\alpha \neq 1.5$ or $\beta \neq 2$, then these distortions are further exacerbated by migrations. The payoffs of $L_1$ are not affected by migrations. For the other two classes, out-migrations benefit a class if it is already making a strategically beneficial choice of its utility function. Out-migrations or in-migrations have a small impact on the net payoffs of non-participating classes. The observations can be summarized in Table \ref{tab: soapmarketYtab}.
\begin{table}[H]
    \centering
$\begin{array}{|c|c|c|c|}\hline 
class & regime & Y_{21} & Y_{32} \\ \hline \hline 
L_2 & \alpha <1.5 &\uparrow & \downarrow \\ \hline 
L_2 & \alpha >1.5 & \downarrow & \uparrow \\ \hline 
L_3 & \beta <2 & - & \uparrow \\ \hline 
L_3 & \beta >2 & - & \downarrow \\ \hline 
\end{array} $
    \caption{Increase/decrease in the payoffs of strategic classes, with the  migration effect}
    \label{tab: soapmarketYtab}
\end{table}

\section{Conclusion} \label{sec: conclusions}

This paper shows that consumer choice is indeed an important determinant of the wage distribution in an economy. This connection provides an important tool for wage-earners to understand how they can adapt their consumption so as to support a more equitable distribution of wages. It does this by providing a modeling and analytic framework which allows us to explore concretely the thread between consumer choice, prices, production, and wages.  

The paper also helps us understand the pricing of many everyday items, e.g., smartphones, where two similar devices may have very different prices, and also that these prices may dramatically change based on a fluid consumer choice. It also suggests that preferring goods and services provided by small-branded and local/regional players, rather than buying the "best" may be a better strategy to ensure better wages. 

Next, the key data required of the economy, viz., $T$ and $Y$, is a {\em labor inventory} of the production processes of the economy and is part of some of the standard data sets of countries. Such an inventory could be used to develop a tool allowing each household to compute its {\em labor footprint}, i.e., an understanding of how household consumption brings employment across the economy. One conjecture is that the  consumption preferences of many wage-earners possibly do not support their  own employment. Such an understanding may be useful to these very classes in modifying their personal consumption. Also note that the labor footprint,  while very similar to the GDP calculation, does not need monetization. This is important in its own way.


Technically, the computation of an equilibrium, given a $T,Y,U$ is an interesting problem. The A-D connection implies that already-known efficient algorithms will shed some light. The tatonnement process of this paper needs to be strengthened and its "computing" power needs to be enhanced. A study of the pay-off correspondence and the existence of Nash equilibria in general $(T,Y,U)$ markets need to be undertaken. 

\keywords{economic models, wages, inequality, strategic behavior, pricing, labor markets, consumer choice, utility maximization, Fisher markets, duality, technology, Arrow-Debreu market, equilibria  }
\newpage
\bibliographystyle{ACM-Reference-Format}
\bibliography{biblio}

\appendix
\section{Section \ref{sec:modeling} details}\label{app: sec2details}

\subsection{Production and Consumption sides: Illustrations}
\begin{ex} \label{ex: prodex1}
(Production) Consider the following economy with two labor classes and two goods as follows:
\begin{equation}
T= \left[ \begin{array}{cc} 1 & 0 \\ 5 & 10 \end{array} \right] \: \: \: 
Y = \left[ \begin{array}{c} \delta \\ 10 \end{array} \right] \: \: \: 
\end{equation}
The number of people in labor class is a parameter $\delta $. The price vector $p=[1,1]$. The labor constraints give us the equations:
\begin{equation} \begin{array}{rcl}
q_1 & \leq & \delta \\
5q_1 + 10q_2 &\leq & 10 \end{array} \end{equation}
We see that for $\delta \leq 2$, $q_1 =\delta $ and $q_2 = 1-0.5\delta $, gives us the maximum revenue, which is $p_1 q_1 +p_2 q_2 =1+0.5\delta$. The wages must satisfy $wT=p$ whence they become $w_1 =0.5$ and $w_2 =0.1$. The net wage bill is $0.5\cdot \delta +10\cdot 0.1$ which equals the revenue. 

However for $\delta >2$, we see that $q=[2,0]^T$ and the revenue is fixed at $2$. Also note that $(Tq)_1 =2< \delta $. This forces $w_1 =0$. The wages $w$ must satisfy $wT\geq p$ and since $q_1 >0$, we must have $(wT)_1 =p_1 $. This gives us $w_2 =0.2$. Finally $(wT)_2=2 >p_2=1$ which, of course, is complementary to $q_2 =0$. 
\end{ex}
\begin{ex} \label{ex: prodex2}
(Trade and Production) Consider societies $A$ and $B$ with populations $10$ and $60$ respectively, and with two commodities, viz., $g_1 $, i.e., rice and $g_2 $, i.e., computers. Suppose that in $A$ it takes $4$ units of labor years to produce a bag or rice and $2$ units of labor years to manufacture a computer. For $B$, these numbers are $6$ and $15$ respectively. 
Suppose that the utility function for both societies are the same and are given by family $u(q_1 ,q_2 )=\sqrt{q_1 q_2 }$, where $q_i $ is the production of good $g_i $.

Thus we have:
\begin{equation} T_A = \left[ \begin{array}{cc} 4 & 2 \end{array} \right] \: \: \: 
Y_A = \left[ \begin{array}{c} 10 \end{array} \right] \: \: \: 
T_B = \left[ \begin{array}{cc} 6 & 15 \end{array} \right] \: \: \:
Y_B = \left[ \begin{array}{c} 60 \end{array} \right]
\end{equation}
Thus, the production frontier $P^A$ of society $A$ is given by 
\begin{equation} P^A =\{(q^A_1 , q^A_2 )\: |\:  q^A_1 ,q^A_2 \geq 0 \mbox{ and } 4q^A_1 +2q^B_2 \leq 10\} \end{equation}
while that of $B$ is given by 
\begin{equation} P^B =\{(q^B_1 , q^B_2 )\: |\:  q^B_1 ,q^B_2 \geq 0 \mbox{ and } 6q^A_1 +15q^B_2 \leq 10\} \end{equation}
We can see that the optimal production is given by the vectors $q^A =(1.25,2.5)$ and $q^B =(5,2)$. The joint production frontier for the two societies together is given by:
\begin{equation} T=\left[ \begin{array}{cccc} 
4 & 2 & 0 & 0 \\ 
0 & 0 & 6 & 15 \end{array} \right] \: \: \: 
Y=\left[ \begin{array}{c}
10 \\ 60 \end{array} \right] \end{equation}
The production frontier in terms of rice and computers, both societies taken together, is given by the convex hull of $\{ (0,0),(0,9), (10,5),(12,0)\}$. The optimal production is given by the point $\bar{q}=(10,5)$ which is better than the sum $q^A+q^B$, and is the basis of trade.  
\end{ex}

\begin{ex} \label{ex: conex1}
(Consumption) Here, we illustrate the Fisher market through a small example. Consider two buyers $L_1 $ and $L_2 $, and two commodities $g_1 $ and $g_2 $, with $1$ unit of each available for sale. Let $x_{ij}$ be the allocation of good $j$ to buyer $i$. Suppose that $L_1 $ prefers $g_1 $ twice as much as $g_2 $, while $L_2 $, prefers $g_2 $ thrice as much as $g_1 $. In other words, we may write $U$ as:
\begin{equation} U= \left[ \begin{array}{cc}
2 & 1 \\ 1 & 3 \end{array} \right] \end{equation}
Suppose $W_2 =10$ is fixed. Let us consider first the case when $W_1 =10+\delta$. For $\delta \leq 10$, we claim that the price vector is $\bar{p}(\delta)=(10+\delta,10)$ and that $x_{11}=x_{22}=1$ with $bb_1 =2/(10+\delta)$ and $bb_2 = 3/10$. 
In other words, $L_1 $ consumes $g_1 $ and $L_2 $ consumes $g_2 $. For the price of $g_1$ as $10+\delta$, $g_2 $ remains the most attractive good for $L_2$. Similarly, while $10+\delta =p_1 \leq u_{12}\cdot p_1 =2 \cdot 10$, $g_1 $ remains the most attractive good for $L_1 $. Moreover, as $\delta $ increases, the price of $g_1 $ rises and absorbs all the money of $L_1 $. Thus the prices $p$ are $p=(10+\delta,10)$ for $\delta \leq 10$ and the Fisher graph remains fixed at $\{ (1,1),(2,2)\}$. The net revenue is $20+\delta $, i.e., the net money with the buyers.  


Things change when $\delta > 10$, and $L_1 $ starts consuming $g_2 $ as well. Thus $x_{11}=1$ and $x_{12}>0$ while $0<x_{22}<1$, and the Fisher graph now becomes $\{ (1,1),(1,2),(2,2)\}$. By the {\em optimal goods} condition for buyer $1$, we have $p_2 /u_{12}=p_1 /u_{11}$. This fixes the price of $g_1 $ as twice that of the price of $g_2 $, since $L_1 $ is consuming both. Thus, if the price of $g_2$ if $p_2$, then the price of $g_1 $ is $2p_2 $. The {\em market clearing} condition gives us two equations:
\begin{equation} \begin{array}{rcl}
10+\delta &=& 2p_2+x_{12}\cdot p_2 \\
10&=& (1-x_{12})\cdot p_2 \end{array}\end{equation}
Solving for $p_2$ and $x_{12}$ gives us:
\begin{equation} \begin{array}{rcl}
p_2&=& \frac{\delta+20}{3}\\
x_{12}&=& \frac{\delta -10}{\delta+20} \end{array} \end{equation}
The prices $p$ are $p=(\frac{2\delta+40}{3}, \frac{\delta+20}{3})$. The net revenue is as before, i.e., $20+\delta $. 
\end{ex}

\subsection{SPLC utilites} \label{app: splc}
An important extension of the linear utility Fisher market is the separable piecewise-linear concave (splc) utility Fisher market. Here, the utility function of each agent is given by an splc continuous function $u_i :(\R^+)^n \rightarrow \R^+$, where $u_i (\bar{x})$ is the utility of the consumption $\bar{x}=(x_1 ,\ldots , x_n)$ to agent $i$, and which in turn is given by 
\begin{equation} u_i (\bar{x})=\sum_j u_{ij} (x_j ) \end{equation}
where each $u_{ij}:\R^+ \rightarrow \R^+$ is splc. For this case, the optimality condition is revised as follows: 

\noindent
\textbf{Optimal Goods:} For every $x_{ij}>0$, let $I_{ij}$ be the interval of slopes of the function $u_{ij}$ at the point $x_{ij}$, i.e. $I_{ij}=[u_{ij}'(x_{ij}^+),u_{ij}'(x_{ij}^-)]$.  Let $BB_{ij}$ = $I_{ij}/p_j$ (whenever $I_{ij}\neq [0,0]$), the interval of possible bangs-per-buck. Note that the interval $I_{ij}$ is a point if $x_{ij}$ lies in the interior of a particular linear segment of the function $u_{ij}$. It is a proper interval only when $x_{ij}$ is on the intersection of two linear segments. The optimality condition then is the requirement that $\cap_{j:x_{ij}>0} BB_{ij} \neq \phi$. This corresponds to the existence of a common bang-per-buck value for all consumption by agent $i$. 

\subsection{The Analysis of \emph{ Tatonnement}} \label{app: tatonnement}
A detailed description of the \emph{tatonnement} iterator function, as described in Section \ref{sec:modeling}, is given below.
\begin{enumerate}
\item Input $\eta_0 =(p_0 , q_0, w_0 ,X_0 )$. Put $n=0$.
\item We first check if the state  $(p_n, q_n, w_n, X_n)$ is an equilibrium, by checking if $(q_n,w_n)\in CP(T,Y;p)$ and $(p_n,X)\in CC(U;q,w)$. If so, then $\eta_n $ is an SM equilibrium. If not, we follow the iterative steps below.  
\item Using $p_n$, we first compute $(q_{n+1},w_{n+1})$ by optimizing $f_{\mathcal{P}}$. This is the $n$-th production-side update.
\item We next find $(p_{n+1},X_{n+1})$ through the process $f_{\mathcal{C}}$ using the input $(q_{n+1},W_{n+1})$.  
\item Note that $f_{\mathcal{C}}$ does not set prices for goods not produced. These are set, assuming that a small $\epsilon$ is indeed produced and predicting its price. Thus if $b_1 , \ldots,  b_k$ are the maximum bang per buck values for the players, then 
\begin{equation} p_j = \max_{i}  \frac {u_{ij}} {b_i} \end{equation} 
This tells us that when these Fisher-like prices are offered, for (at least) one player, the maximum bang-per-buck ratio equals the ratio these prices give, making the player buy the good. The computation of $p$ as before and its modification is called $\mathcal{C}(n)$, i.e., the $n$-th consumption-side update. 
\item This completes the description of $\eta_{n+1}$. We go back to Step 2. 
\end{enumerate}
We illustrate an example of SM equilibrium reached through the \emph{tatonnement}.
\begin{ex}
Let us  consider a market with 3 classes and 3 goods with the following specifications for technology, utility, and labor availability.\\
\begin{equation} T=\left[ \begin{array}{ccc}
1 & 0 & 2 \\ 
3 & 4 & 0 \\
0.5 & 2.5 & 2 \end{array} \right] ; \:
U=\left[ \begin{array}{ccc}
1.5 & 0.41 & 0 \\ 
0.58 & 1.1 & 0.2  \\
0.5 & 1.4 & 0.6 \end{array} \right] ; \: 
Y=\left[ \begin{array}{c}
1  \\ 
1 \\
1 \end{array} \right] \: \end{equation}
 Starting with the price vector $[ 0.7379, \ \       0.9379,   \ \ 0.3617 ] $, the tatonnement process converges to an equilibrium point  in 3 iterations. These are represented as rows in the following matrices to show the corresponding prices, production, and wages in each iteration:
 \begin{equation}\left[ \begin{array}{c}
p_1\\
     p_2 \\
     p_3 \\
     p_4
\end{array} \right] = \left[ \begin{array}{ccc}
       2.0408   &    3.8706   &    0.7037 \\
       1.5961    &    3.027    &   1.2973 \\
       1.5099     &  2.8636     &  1.2273 \\
       1.5099      & 2.8636      & 1.2273 
\end{array} \right], \ \left[ \begin{array}{c}
q_1\\
     q_2 \\
     q_3 \\
     q_4
\end{array} \right] =  \left[ \begin{array}{ccc}
0.2632   &    0.0526   &    0.3684\\
            0     &    0.25   &    0.1875 \\
       0.2631      & 0.0526    &   0.3684\\
       0.2631   &    0.0526     &  0.3684  \end{array} \right] ; \end{equation}
\begin{equation} \left[ \begin{array}{c}
w_1\\
     w_2 \\
     w_3 \\
     w_4
\end{array} \right] =\left[ \begin{array}{ccc}
 0.31636    &  0.52007   &   0.16357 \\
            0    &     0.68 &        0.32 \\
     0.086693     &  0.3865  &    0.52681\\
     0.086693      & 0.3865   &   0.52681\\ \end{array} \right]  \end{equation}
 And the allocations $X$ are given by the forest-
 \begin{center}
\includegraphics[width=2.8cm]{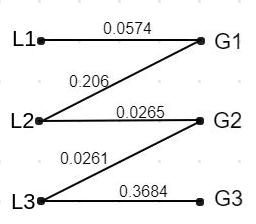}
\end{center}
\end{ex}
\noindent Next, we provide an example where the process fails to converge.
\begin{ex} \label{app: tatonnementfails}
Let us now consider a market with  3 classes and 3 goods with the following specifications.\\
\begin{equation} T=\left[ \begin{array}{ccc}
0.05 & 1 & 0.9 \\ 
0.5 & 0.8 & 0.15 \\
0.4 & 0.5 & 0.4 \end{array} \right] ; \:
U=\left[ \begin{array}{ccc}
0.2 & 0.3 & 0.8 \\ 
0.9 & 0.2 & 0.4 \\
0.25 & 0.85 & 0.33 \end{array} \right] ; \: 
Y=\left[ \begin{array}{c}
10  \\ 
10 \\
10 \end{array} \right] \: \end{equation}
Tatonnement does not converge here, instead, the following two states are reached that stay in a loop -
\begin{equation} q=\left[ \begin{array}{ccc}
4.35 & 9.78 & 0 \\ 
14.7 & 0 & 10.3 \end{array} \right] ; \:
p= \left[ \begin{array}{ccc}
0.11  &  0.05  &  0.14 \\
 0.04 &   0.12  &  0.05 \end{array} \right] ; \: 
w =\left[ \begin{array}{ccc}
    0.52  &  0.48    &     0 \\
    0.12 &        0  &  0.88 \end{array} \right]  \end{equation}

\end{ex}

 In Example \ref{app: tatonnementfails}, two states toggle cyclically. In this case, it is not possible for a state to exist with all goods and classes simultaneously active. The solution involves two states, with class-1 and good-1 active in both the states and other goods and classes alternate between the two.  In general, we see the following necessary condition for an equilibrium: For all goods ($j$) that are not active in the market, 
\begin{equation} p_j = \max_i \frac{u_{ij}}{b_i} \text{ (U- Level}) < \sum_i T_{ij}w_i \text{ (T- Level})\end{equation}
As we see earlier, if a good is not produced, it is allotted a Fisher-like price in the \emph{tatonnement}, which is the $U$-level defined above. All produced goods have prices greater than or equal to their $U$ levels, as guaranteed by the Fisher market. The $T$-level, on the other hand,  is defined as the amount of money required to produce a good. It is then clear that if $U$-Level for an unproduced good $j$ is more than its $T$-Level i.e. $p_j > (w \cdot T)_j$, then being profitable,  good $j$ becomes active in the next iteration, by perhaps pushing a less efficient good $j'$ out of production. If subsequent iterations assign good $j'$ a $U$-level price larger than its $T$-level, then it may become active again. 

As seen from the examples, this method does work similarly to the {\em tatonnement} process given in  Walrus' theory of general equilibrium \cite{uzawa1960walras}. It too starts with a price vector, computes production and wages, and produces the next set of prices based on these market variables. It is clear that the process terminates if and only if it attains an equilibrium. From Example \ref{app: tatonnementfails}, it can be observed that the process may not always converge, and there may be toggling states. Moreover, it can be observed that equilibria whose Fisher forests are disconnected are unlikely to arise from the above process, even though they are fixed points. 

\section{Section \ref{sec: ccg} details}\label{app: sec4details}
\subsection{Proof of Theorem \ref{thm: T1}} \label{app: T1proof}
\thmone*
 \begin{proof}
Let us consider an economy with $m$ labor classes and $n$ goods.  Without loss of generality, we examine the existence of equilibrium so that all labor classes and goods are active. Let us  $n-m$ as the deficit $`def$'. We first consider the case $k-1 < def$ and prove that no equilibrium can exist in this case. We then consider the case $k-1 \geq def$ and state the possibilities of equilibrium by giving a system of solvable equations. With this, we prove that for any generic or non-generic equilibrium, it is possible to find an arbitrarily close  generic equilibrium by changing the control variables i.e $U'$. We conclude the proof by giving an example of a $2 \times 3$ case. 

\textbf{ $k-1 < def$}:
Let us assume that there exist $p,q,w$ such that $p= wT$, $Tq=Y$ and $p$ is the Fisher market solution for the market set up by $U, q$ and $w$ i.e. there exists an SM  equilibrium ($p,q,w, X$). Using the prices, it is possible to generate the Fisher forest. We assume that there are $k$ trees in the forest so that $k-1 < def$. \\ Setting up the profitability conditions, we have, 
\begin{equation}   \left[  \begin{array}{cccccc}
   T_{11}  & T_{21} & ...& T_{m1}  & | & p_1 \\
     T_{12}  & T_{22} & ...& T_{m2}  & | & p_2 \\
     . & . & . & . & | & . \\ 
          . & . & . & . & | &  . \\ 
T_{1n}  & T_{2n} & ...& T_{mn}  & | & p_n
\end{array} \right]  \left[  \begin{array}{c}
     w_1 \\
     w_2 \\
     . \\
     .\\
     w_m \\
     -1
\end{array} \right] = \left[ \begin{array}{c}
     0  \\
     0 \\
     . \\
     . \\
     0
\end{array} \right] \end{equation}
We have $k \leq def$.  Let us assume that the goods $g_1 , \ldots ,g_k $ are in distinct components and  for the rows $R_1 ,\ldots , R_k \in \R^n$, where $R_i (i)=1$, $R_i (j)=0$ if $g_j $ is not in the component of $g_i $, and finally $R_i (j)=\alpha_{ij}$ a monomial in the entries of $U$ which relates the price of $g_i $ with $g_j $. If we are to assume that $p_1 ,\ldots , p_k $ as temporarily known, then the equation $wT=p$ leads us to the equation:
\begin{equation} \left[ \begin{array}{cccc} p_1 & \ldots & p_k & -w \end{array}\right]\left[ \begin{array}{c}
R_1 \\
\vdots \\
R_k \\
T \end{array} \right] =0  \label{equation} \end{equation}
Thus, in the case of $k= def$, we can rearrange the matrix so that the determinant of the following matrix is zero. 

\begin{align}
 \left[ \begin{array}{cccccccc}
   T_{11}  & T_{21} & ...& T_{m1}  & | & p_1 & . & .\\
     T_{12}  & T_{22} & ...& T_{m2}  & | & \alpha p_1 & . & . \\
     . & . & . & . & | & .  & p_2 & . \\ 
     . & . & . & . & | &  . & \beta p_2 & .\\ 
          . & . & . & . & | &  . & . & .\\ 
          . & . & . & . & | &  . & . & .\\ 
T_{1n}  & T_{2n} & ...& T_{mn}  & | & . & . & \gamma p_k 
\end{array} \right]
\end{align}

Here, $\alpha, \beta.., \gamma$ represent the algebraic relationships between prices of goods belonging to the same components of the forest. Each column on the right corresponds to a tree. We have rearranged $p$ so that there are $k$ columns. It is clear that the $p$'s can be taken out from the determinant so that the condition translates to an algebraic relationship between $\alpha, \beta.., \gamma$ i.e.  ratio of entries of $U$ and coefficients of $T$. This contradicts the assumption that $U$ and $T$ are independent. In case $k < def$, we can delete $def-k$ number of rows from the above matrix so that it becomes a $(m+k) \times (m+k) $  matrix. The determinant of this matrix is zero, which again contradicts the assumption. Thus, there does not exist any equilibrium when $k-1 < def$.

\textbf{$k-1 \geq def$}:
We again consider (\ref{equation}), which is a system of $n$ homogeneous equations in $m+k$ unknowns. Thus, if the first $m+k-n=k'$ prices were fixed, then all prices and wages are defined using these $k'$ prices. Here, $k' \geq 1$.

Next, let us consider the conservation of money in each of the $k$ components $G_i$ and use these to solve for $q$. We of course have $wT=p$, whence $wY=wTq=pq$. Thus the overall conservation of money is already available and there are only $k-1$ independent equations. Next, We have the $m$ equations $Tq=Y$. Now, for each component, the money available is given as $\sum_{L_j \in G_i} Y_j w_j $. On the goods side, we have $\sum_{g_j \in G_i } p_j q_j$. 
By the earlier argument, these translate into $k$ algebraic equations in the variables $p_1 ,\ldots ,p_{k'}$ and $q_1 ,\ldots ,q_n$, which are homogeneous in $p$'s. By dividing throughout by $p_1 $, we get $k-1$ equations in $k'-1$ variables. This gives us a total of $m+k-1$ equations in $n+k'-1$ variables. Substituting $k'=m+k-n$ gives us $m+k$ equations in $m+k-1$ variables. These may be solved to obtain all quantities.  Section-6 shows the existence of a disconnected $(2 \times 2$ market i.e. $k-1 = 1 > def =0$, which confirms that $k-1\geq def$.

Let us now consider the case where this equilibrium is generic. We recall the definition of generic equilibrium points- we say, $\eta $ is a generic equilibrium point if (i) for $j\not \in J(\eta )$, we have $(wT)_j >p_j $, and (ii) for $(i,j) \not \in F(\eta )$, we have  $u_{ij}/p_j < \max_k u_{ik}/p_k $. In this case, the set of tight equations is the same as the set given by $(I,J,F)$ i.e. $Tq=Y$, $wT= p$, and the maximum bang per buck and money conservation constraints given by the Fisher market. Therefore, it is possible to change a $u_{ij}$ by $\delta$ and retain the equations- $(i,j) \not \in F(\eta )$, \   $u_{ij}/p_j < \max_k u_{ik}/p_k $, which result in a generic equilibrium point. It is clear that the equations $(wT)_j >p_j $ hold true. On the other hand, when we have a non-generic point,  i.e. $(wT)_j =p_j $ for $j\not \in J(\eta )$ or $u_{ij}/p_j = \max_k u_{ik}/p_k $ for $(i,j) \not \in F(\eta )$, the number of tight equations is more than those given by its $I,J,F$.  We can then  relax an equation to reach a generic point.
  
In all, we have proved that   no equilibrium exists if  $k < def+1$.  Moreover, we can see that the condition $k-1 \geq def$ is necessary but not sufficient, as there are more conditions like the nonnegativity of variables and the optimality of prices. We also note that the condition of a generic equilibrium point is crucial. To see this, consider a non-generic equilibrium point with $(i,j)$ satisfying the bang per buck condition for buyer $i$ such that $x_{ij} = 0$, i.e.,  $(i,j) \not \in F(\eta)$. This increases the number of actually connected components, thereby adding an extra tight constraint, which violates the generic economy argument. Thus, if there are $z$ such zero-weight edges, the number of connected components is at least $def+1+z$.  We now complete the proof showing a $2 \times 3$ case in Example \ref{ex: 2n3case}. 
\begin{ex}\label{ex: 2n3case}
Let us consider an economy with the following specifications. 
\begin{equation} T =
 \begin{bmatrix}
1 & 0.5 & 0.8 \\
0.5 & 1.5 & 0.9 
\end{bmatrix}  
\ \ \ \ \ U =
 \begin{bmatrix}
    1  &  0.75  &  0.8 \\
    0.4 &  0.9   & 0.7
\end{bmatrix}  
\ \ \ \ \ Y=
\begin{bmatrix}
    1 \\
    1
\end{bmatrix}
\end{equation} 
Here, two labor classes, with one labor unit each, can produce three goods using the above technology constraints. 
Let us consider $p, q, w$ as follows. 
 \begin{equation} p, q, w^T  = \left[ \begin{array}{c}
 0.5235    \\
    0.8324  \\
    0.6470 \end{array} \right], \left[ \begin{array}{c}
      0.5639\\  0.2426 \\ 0.3934 \end{array} \right], \left[ \begin{array}{c}
      0.2952   \\ 0.4565 \end{array} \right] \end{equation}
We see that  these variables follow the equations- 
\begin{equation} p^T = \left[ \begin{array}{ccc}
 0.5235      &     0.8324      &    0.6470 \end{array} \right] \: = 
 w^T \cdot T = \left[ \begin{array}{cc}
      0.2952   & 0.4565 \end{array} \right] \cdot 
      \left[ \begin{array}{ccc} 1 & 0.5 & 0.8 \\
0.5 & 1.5 & 0.9 \end{array} \right] \end{equation}
And 
\begin{equation} Y = \left[ \begin{array}{c}
      1 \\
    1 \end{array} \right]  \ = \:
      T \cdot q = \ \left[ \begin{array}{ccc} 1 & 0.5 & 0.8 \\
0.5 & 1.5 & 0.9 \end{array} \right] \cdot \left[ \begin{array}{c}
      0.5639 \\  0.2426 \\ 0.3934 \end{array} \right] \end{equation}
According to the definition of equilibrium in the SCM model, we see that $q$ and $w$ are dual variables of each other and follow the market constraints. On the consumption part, when $w$ and $q$ are given as inputs to the Fisher market, the solution  is the price vector $p$ with  the following allocations-
\begin{equation} X = 
  \left[ \begin{array}{ccc}  0.5639    &     0  &  0 \\
0  &  0.2426   & 0.3934  \end{array} \right] \end{equation}
That is, class 1 only consumes good 1 while class 2 buys goods 2 and 3. This shows an equilibrium in a general $m \times n$ economy with an unequal number of active goods and classes, with the number of connected components $= n-m+1 = 2$. \end{ex} \end{proof} 

\subsection{Proof of Theorem \ref{thm: genericeq}} \label{app: genericeqproof}
\thmtwo*
 \begin{proof}
  The combinatorial data does give us the relationships $w_I T_{I,J}=p_J$ and $T_{I,J}q_J =Y_I$. From this, it follows that $|I|\leq |J|$ for otherwise there would be an algebraic relationship between $T$ and $Y$. However, if $|I|=|J|$, and $k=1$, then $w$ is determined by $p$ and $q$ by $Y$. Since the forest $F$ is connected, $p$ is determined up to a scalar multiple and thus the whole system is solved. In summary, if $|I|=|J|$ and $k=1$, there is a unique $\eta $ sitting above this combinatorial data. However, in the general case, we must first append to the variables $q_J $, a suitable subset $\{ p_1 ,\ldots ,p_{k'}\}$. The $w$'s and the remaining $p$'s are expressible as homogeneous linear combinations of these $k'$ prices. Next, to the linear set of equations $T_{I,J}q_J =Y_I$ we add the $k-1$ independent money conservation equations to solve these simultaneously. Unfortunately, the conservation equations involve terms $p_i q_i $'s and are quadratic in the chosen variables with coefficients in the entries of $U$. Once these are solved, all other variables are known and the equilibrium point is reconstructed. Thus, over a given combinatorial data, we get an algebraic system with coefficients in $U$, but with finitely many solutions. By standard algebraic geometry results, other than over a closed algebraic set, these solutions depend smoothly on the entries of $U$. 
  \end{proof}
\section{Section \ref{sec: correspondence} details}
\subsection{ Correspondence and the Nash Equilibria in General $2\times 2$ markets} \label{app: stratsinccg}
We now consider a general scenario involving 2 classes and 2 goods where $T, Y$ are given as constants. Here,  given $T$ and $Y$, we rigorously solve for markets, and compute the conditions for a forest to be in equilibrium. We argue that wages are continuous functions of utilities and payoffs and allocations are continuous for each forest. Moreover, we generalize the results given in section 6 by proving the existence of a correspondence between the strategy space and utilities and also that utilities on the boundaries are linear combinations of those of the forests on both sides. We also look at the necessary conditions for Nash equilibria to exist.\\
\noindent
\textbf{Market and specifications: }\\
Let us consider a $2 \times 2$ market with inputs $T$, $Y$. Let $U^t$ be the true utility matrix. Let $U$ stand for the strategy matrix with $\frac{u_{11}}{u_{12}} = \alpha$,  $\frac{u_{21}}{u_{22}} = \beta$.

Let $p_1/p_2$ = $\alpha$ where $\alpha$ is the appropriate ratio of utilities when the forest is connected or the ratio of some technology inputs when the forest is disconnected. For example, if class-1 buys good 1 and class-2 buys good-2, we have $p_2T_{21}^{-1} = p_1T_{12}^{-1}$.  Assuming total money in the economy constant and equal to 1, we can solve for $q$ in terms of $T,Y$.
\begin{align}q_1 = T^{-1}_{11} Y_1 + T^{-1}_{12}Y_2\\
 q_2 = T^{-1}_{21} Y_1 + T^{-1}_{22}Y_2 \end{align}
We assume that $T$ and $Y$ are such that $q_1, q_2$ are positive. Solving for wages, we get
\begin{align}
w_1 = \frac{Y_1(\alpha T_{11}^{-1} + T_{21}^{-1})}{Y_1(\alpha T_{11}^{-1} + T_{21}^{-1}) + Y_2(\alpha T_{12}^{-1} + T_{22}^{-1})}\\
w_2 = \frac{Y_2(\alpha T_{12}^{-1} + T_{22}^{-1})}{Y_1(\alpha T_{11}^{-1} + T_{21}^{-1}) + Y_2(\alpha T_{12}^{-1} + T_{22}^{-1})}\end{align}
Like the case described before, we know that for a $2 \times 2$ market, there are six possible forests and one cycle. Also, there are two possible $(1 \times 1)$ markets.  The possible $(2 \times 2)$ forests are-
\begin{center}
\includegraphics[width=8cm]{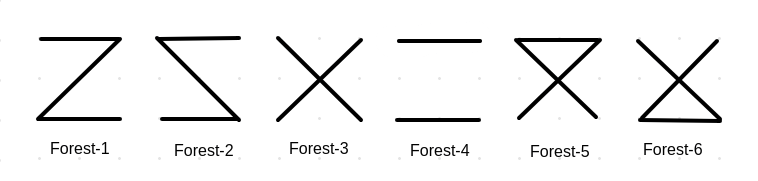}
\end{center}
where each is a solution under specific conditions, which place it in a specific zone. For both classes to be active, we first require $\frac{p_1}{p_2} > \min\left( \frac{-T_{21}^{-1}}{T_{11}^{-1}}, \frac{-T_{22}^{-1}}{T_{12}^{-1}} \right) = \min \left( \frac{T_{21}}{T_{22}}, \  \frac{T_{11}}{T_{12}} \right) $.  For connected forests such that class $i1$ buys only good $j1$ and class $i2$ buys $j1, j2$, we have $w_{i1} < p_{j1}q_{j1}$ and $w_{i2} > p_{j2}q_{j2}$. This ensures that all allocations are positive. Lastly, we require the bang per buck conditions to hold.  \\
For example, consider the first graph given above, where $p_1/ p_2 = u_{21}/ u_{22}= \beta$. This graph constitutes  a generic equilibrium  if and only if, 
\begin{enumerate}
    \item $\frac{p_1}{p_2} = \beta >  \frac{T_{21}}{T_{22}}, \  \frac{T_{11}}{T_{12}}$ i.e. wages are positive. 
    \item $\beta > \frac{T_{21}^{-1}Y_1}{T_{12}^{-1}Y_2}$  i.e. allocations are positive. \\
    $w_1 = Y_1 ( p_1T^{-1}_{11} + p_2T^{-1}_{21}) < p_1q_1 = p_1 (T_{11}^{-1}Y_1 + T_{12}^{-1}Y_2)$  or \\
    $w_2 = Y_2 ( p_1T^{-1}_{12} + p_2T^{-1}_{22}) >  p_2q_2 = p_2 (T_{21}^{-1}Y_1 + T_{22}^{-1}Y_2)$ 
    \item And, $\beta < \alpha = \frac{u_{11}}{u_{12}}$ i.e. class-1 gets more bang per buck from good-1 than good-2. 
    $\left( \frac{u_{11}}{p_1} > \frac{u_{12}}{p_2} \right)$. 
\end{enumerate}
It can be seen that these conditions define an open set in the strategy space $\mathcal{U}$ which has a unique combinatorial data $(I,J,F)$. In the above example, it is given by then $I = \{1,2\}, J= \{1,2\}$ and $F$ as forest-1.  We note that the third condition can be relaxed i.e. $\beta \leq \alpha$, which results in a non generic equilibrium at $\alpha = \beta$. This is because multiple allocations, including the forest-1 allocations are possible for this condition.  We discuss  the cycle i.e. the set $\alpha = \beta$  in the next section.

In addition  to the forests described before, we note that it is possible for wages of a particular class to become zero. When $ p_1/p_2 \leq \min \left( \frac{T_{21}}{T_{22}}, \  \frac{T_{11}}{T_{12}} \right) $, only one class is active and it produces exactly one good. We refer to its corresponding domain in the $U$-space as Zone 7. In this case, if class $i$ produces good $j$, it is possible to compute its production and the utility $U_i$ that it gives. We note that $U_i$ is smaller than the total utility of both classes combined, as there is more net production when both classes are active.\\
\noindent
\textbf{Continuity analysis:}\\
 For the forests described above, the combinatorial  data  $(I,J,K)$ is defined by the following conditions 
\begin{table}[h]
\centering
\begin{tabular}{|c|c|c|c|c|c|}
\hline
 Forest 1 & Forest 2 & Forest 3 & Forest 4  & Forest 5 & Forest 6  \\ \hline
$\alpha \geq \beta = \frac{p_1}{p_2}$ & $\frac{p_1}{p_2}=\alpha \geq \beta$ & $\beta \geq \frac{p_1}{p_2}= \frac{T_{22}^{-1}Y_2}{T_{11}^{-1}Y_1} \geq \alpha$ & $\beta \leq \frac{p_1}{p_2}= \frac{T_{21}^{-1}Y_1}{T_{12}^{-1}Y_2} \leq \alpha$ & $\beta \geq \alpha = \frac{p_1}{p_2}$ & $ \frac{p_1}{p_2} = \beta \geq \alpha$ \\ \hline 
 $\beta >  \frac{T_{21}}{T_{22}}, \frac{T_{11}}{T_{12}}$ & $\alpha >  \frac{T_{21}}{T_{22}}, \frac{T_{11}}{T_{12}}$ & $\frac{T_{22}^{-1}}{T_{11}^{-1}} >  \frac{T_{21}}{T_{22}}, \frac{T_{11}}{T_{12}}$ & $ \frac{T_{21}^{-1}}{T_{12}^{-1}} >  \frac{T_{21}}{T_{22}}, \frac{T_{11}}{T_{12}}$ & $\alpha >  \frac{T_{21}}{T_{22}}, \frac{T_{11}}{T_{12}}$ & $\beta >  \frac{T_{21}}{T_{22}}, \frac{T_{11}}{T_{12}}$ \\ \hline
$\beta > \frac{T_{21}^{-1} Y_1}{T_{12}^{-1}Y_2}$  & $\alpha < \frac{T_{21}^{-1}Y_1}{T_{12}^{-1}Y_2}$ &  & & $\alpha > \frac{T_{22}^{-1} Y_2}{T_{11}^{-1}Y_1}$ & $\beta < \frac{T_{22}^{-1} Y_2}{T_{11}^{-1}Y_1}$ \\ \hline
\end{tabular}
\end{table}

We also graphically represent these six forests by marking their zones. We number these forests from 1 to 6, ordered as given above. 7th zone refers to the case where only one class is active and in equilibrium. Here, the arrows indicate the direction of optimization for the classes, considering $U^t$ as the strategy matrix comprising of all ones.
\begin{center}
\includegraphics[width=9cm]{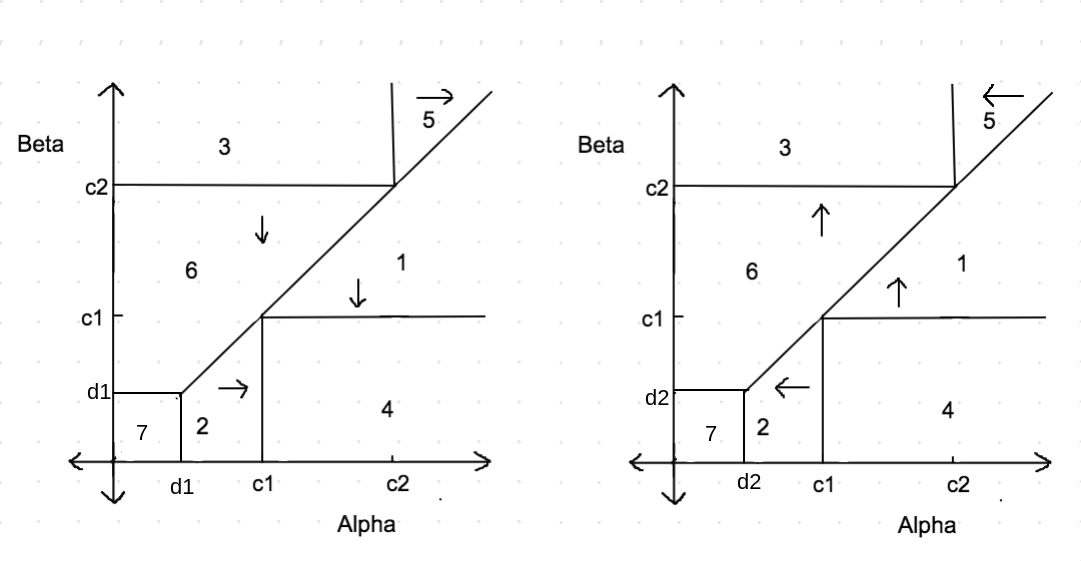}
\end{center}

 Here, $c1 = \frac{T_{21}^{-1} Y_1}{T_{12}^{-1}Y_2}$ and $c2= \frac{T_{22}^{-1} Y_1}{T_{11}^{-1}Y_1}$ and we have assumed $c2>c1$. The first graph refers to the condition $\det (T) > 0 $ i.e. $d1 = \frac{T_{11}}{T_{12}} > d2 = \frac{T_{21}}{T_{22}}$ and second refers to $\det (T)< 0$.  There are four possible schemes corresponding to $c1, c2$ and $\det (T)$, of which we have considered two with $c2>c1$.  In the first scheme, class-1 always gets wages and vice versa.    Note that the region for the 7th zone can extend depending on $d1$. If $d1 > c1$, as we will see, the continuity of wages implies zone 4 becomes a part of zone 7. 

We see that there are clear ranges of $\alpha$ and $\beta$ which define the solution forest. First, we observe that as wages are rational functions involving only $p1/p2$ i.e. the relevant $\alpha, \beta$ in case the forest is connected, wages are always continuous within a forest. For disconnected forests, wages are constant. It is important to note that wages are bounded, $w_1+w_2 = 1$. At the boundary of zone 4, $p1/p2$ i.e. $\alpha$ or $\beta$ equals $c1$, which determines the wages. In other words,  as $\alpha$ or $\beta$ continuously approach $c1$, limit of $w_i$ equals $w_i$ for zone 4.  Similarly, zone 3 allows for smooth transitions. It is clear that wages are continuous at the boundary $\alpha = \beta $ also, as one relevant variable smoothly transitions  into another. Similarly, in zone 7, wages continuously converge to zero.

Next, we observe that allocations and utilities are continuous within each forest, where they are functions. For disconnected forests, they are constants. For connected forests, where there are three edges,  we see that 2 allocations are linearly dependent on $p_1/p_2$ and one is constant. For example, forest 1 has the following allocations- 
\begin{align}
x_{11} = w_1/p_1 = Y_1T_{11}^{-1}  + Y_1 T_{21}^{-1}/\beta\\
x_{21} = q_1- x_{11},   \ \ x_{22} = q_2
\end{align}Since utilities (true and strategic) are linear functions of allocations, those are continuous too. Moreover, they are either linear or rational functions of $\alpha, \beta$. Though we have considered a special strategy matrix with all ones here, the result is general.\\
\noindent
\textbf{Transitions and the cycle:}\\
 Firstly, we  note that except for the $\alpha= \beta$ line, other boundaries allow for a continuous function where for each pair of ($\alpha, \beta$), a unique allocation/utility exists.  Let us see this by changing $\alpha$ or $\beta$ continuously to approach a forest/cycle from another. We see that when a tree is disconnected by making an edge weight zero, allocations and utilities continuously transform. For example, in zone -1,  when $\beta$ approaches $c1$ from above, $p_1/p_2$ equals $c1$ where zone-4 starts, and $x_{11}$ becomes $q_1$.  Therefore, forest 3,4 include their boundaries while defining utility functions i.e. boundaries have same utilities as the forests.  In case of $\alpha=\beta$, first an edge is added to get a cycle, where multiple allocations are possible and then another edge is removed. 
 
 The governing equations for a connected forest, say forest 1 are : 
 $w_1 = x  < p_1q_1$, $y < p_1q_1$ and $z = p_2q_2$ along with $x +y = p_1q_1$. When we increase $\beta$ to approach the cycle, the equations become $x, y < p_1q_1$, and $z,v< p_2q_2$ along with $x+y = p_1q_1$ and $z+v =p_2q_2$. Thus, the maximum value of $z$ is $p_2q_2$ where zone 1 is  overlapping  and minimum value is 0 where $v = p_2q_2$ when zone 6 overlaps. 
 \begin{center}
\includegraphics[width=8cm]{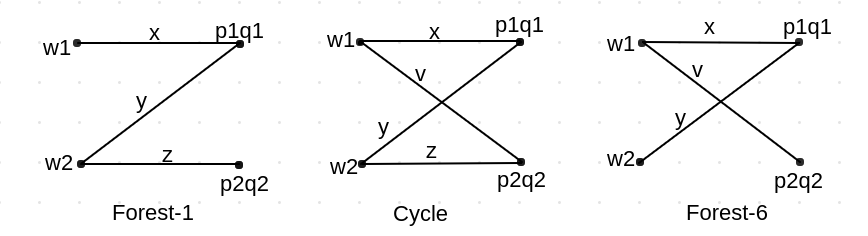}
\end{center}
 The net utilities for forest 1 and 6 are-
 \begin{align}
 (b_1,b_2)_1  = \left(\frac{w_1}{p_1} , \ \ q_1+q_2 - \frac{w_1}{p_1}\right)\\
(b_1,b_2)_6 = \left(q_2 + q_1 - \frac{w_2}{p_1}, \ \ \frac{w_2}{p_1}\right)\end{align}
 And for the cycle, these can be computed in terms of $v,z$-
 \begin{align}(b_1,b_2 ) & = \left(\frac{w_1 - v}{p_1} + \frac{v}{p_2}, \ \ \frac{w_2-x}{p_1} + \frac{z}{p_2}\right)\notag \\
 & = \left(\frac{w_1}{p_1} + v\left(\frac{1}{p_2} - \frac{1}{p_1}\right), \ \ \frac{w_2}{p_1} + z\left(\frac{1}{p_2} - \frac{1}{p_1}\right)\right)\end{align}
 Firstly, we note that the ratio $1/p_2 - 1/p_1$ is a constant for any fixed $\alpha, \beta$. Let this be a positive constant. Now, it is clear that when $v = 0$, $b_1$ for the cycle is minimum which coincides with that of forest-1. Similarly, when $v = p_2q_2$, the maximum coincides with that of forest-6. This establishes that utilities at the cycle are linear combinations of those on the boundaries, so that for each $\alpha$ on the $\alpha = \beta$ line, utilities make a smooth transition. This makes us view the cycle as a linear combination of two forests depending on values of $\alpha$ and $\beta $. Though we have considered $U^t$ made up of all ones, this result is valid for  any matrix $U^t$. 
 
 We see that the results generalize the example given in section 6. Using the same proof, we can now argue that there exists a correspondence $\mathcal{N}$ between the strategy space and utilities. Also, continuous payoffs for the forests and the range for the line $\alpha = \beta$ allow for a 2 dimensional manifold whose boundaries are given by those of the forests. Next, we look at the conditions for a Nash equilibrium to exist. \\
 \noindent
\textbf{Existence of a Nash equilibrium: } \\
We observe that allocations are functions of these variables, either linear or inversely related ( $U_i = a+ b \alpha $ or $c + d / \alpha $), when the forests are connected.  In those cases, maximization can occur only at the boundaries i.e. at the points where the transition of graphs occurs. When the forests are disconnected, allocations and utilities are constants for a range of $\alpha, \beta$. Therefore, if any Nash equilibrium exists, it should be at the transition points or forest 3 or 4. 

Let us now analyze the first scheme. Looking at the arrows, we know that forest-3 cannot be an equilibrium, which leaves $\alpha= \beta$ and forest-4 as possibilities.  Class-2 prefers forest-4 to forest-2, and similarly, class-1 prefers forest-4. For a point in forest-4 to be an equilibrium, $\beta > d1$ and $\alpha > c1$. In this region, given any $\alpha$, class-2 prefers forest-4 to forest-1,3 and 5. We now have to find a region for $\beta$ so that $\alpha$ maximises $b_1$ in that region. Note that when $c1 > \beta > d1$, $b_1$ is constant as a function of $\alpha$ in zone 6. When it enters zone-2, $b_1$ increases until it is in zone 4. Thus, we prove Nash equilibria exist in the region $d1 < \beta < c1$ and $ c1 < \alpha $. Similarly, it can be proved that Nash equilibria exists in scheme-2 in zone-3 for $d2 <\alpha < c1$ and $c2 < \beta$. This equilibrium is, however, subject to the condition that the forest exists i.e. the region is feasible. For example, in scheme-1, forest 4 exists when $d1 < c1$ i.e. $\frac{T_{11}}{T_{12}} < \frac{T_{21}^{-1} Y_1}{T_{12}^{-1}Y_2} $ and $\beta < c1 < \alpha$. 

This also settles the question of the existence and uniqueness of equilibrium in a $2 \times 2$ case. 
\subsection{Description of zones in Table \ref{tab: forestdescriptions}} \label{app: foresttable}
\begin{table}[H]
\centering
\small
\begin{tabular}{|l|l|l|l|l|l|}
\hline &&&&&\\[-0.5em]
 & Forest-1                                                                                                      & Forest-2                                                                                                     & Forest-3         & Forest-4              & Forest-5                                                                                \\ \hline &&&&&\\[-0.5em]
 Prices & $\beta=  \frac {p_1}{p_2}   > 1/4$ & $\alpha= \frac {p_1}{p_2}  > 1/2$ & $1/4 < \beta  =\frac {p_1}{p_2}  < 1/2 $ & $\frac{p_1}{p_2} = \frac{1}{2}$ & $p_1$ = 0\\ 
 \hline &&&&&\\[-0.5em]
 Alloc - 1 & $x_{11} = 8 - \frac{2}{\beta}$ & $x_{11} = 8 - \frac{4}{\alpha}$, \ \ $x_{12}=2$ & $x_{12}= 8 \beta - 2$  & $x_{12} = 2$ & $x_{11},x_{12} = 0 $                     
 \\ \hline &&&&&\\[-0.5em]
  Alloc - 2 & $x_{21} = \frac{2}{\beta} , \ \  x_{22} = 2 $                              & $x_{21} = \frac{4}{\alpha}  $ & $x_{21} = 8 , \ \ x_{22} = 4 - 8\beta$ &   $x_{21} = 8$ & $x_{22}=4$                                 \\
     \hline &&&&&\\[-0.5em] $b_1$   & $ b_1= 8 - 2/ \beta$                                                                                        & $ b_1= 10-4/\alpha $                                                                                   & $ b_1 = 8\beta - 2$  & $b_1 = 2$ & $b_1=0$ \\
     \hline &&&&&\\[-0.5em] $b_2$ &
     $ b_2= 2+ 2/\beta $                                                                                      & $ b_2= 4/\alpha $                                                                                             & $ b_2=    12-8\beta$  & $b_2 =8 $     & $b_2 = 4$                                                                \\
     \hline &&&&&\\[-0.3em]
     Wages-1  & $w_1 = 1- \frac{2}{1+4 \beta}$ & $w_1=1 - \frac{2}{1+4 \alpha}$ & $w_1 = 1- \frac{2} {1+4 \beta} $ & $w_1 = 1/3$ & $w_1=0$ \\
     \hline &&&&&\\[-0.3em]
     Wages-2  & $w_2 = \frac{2}{1+4 \beta}$ & $w_2= \frac{2}{1+4 \alpha}$ & $w_2 = \frac{2} {1+4 \beta} $ & $w_2 = 2/3$ & $w_2=1$ \\
     \hline &&&&&\\[-0.3em]
     Zones  & $\alpha \geq \beta > 1/4$  & $\beta \geq \alpha > 1/2$ & $1/4 < \beta< 1/2$ and  $\beta \geq \alpha$ & $ \alpha \leq 1/2 \leq \beta $ & $\beta \leq 1/4$
     \\ \hline
\end{tabular}
\caption{$SM(T,Y,U)$ characteristics for the forests in Fig. \ref{fig: ccgexforests}} 
\end{table} 

\end{document}